\documentclass[letterpaper,11pt]{article}
\pdfoutput=1 % if your are submitting a pdflatex (i.e. if you have
\usepackage{jheppub} % for details on the use of the package, please
\usepackage{graphicx}
\usepackage{epstopdf}
\usepackage{amsmath, amssymb, amsthm}

\usepackage{tikz}
\usepackage{verbatim}
\usepackage{bm}
\usepackage{caption}
\usepackage{subcaption}
\usepackage{csquotes}

\usepackage{ytableau}% for including young's tableau in equations

\usepackage[us,12hr]{datetime}

\newcommand{\be}{\begin{eqnarray}}
\newcommand{\ee}{\end{eqnarray}}
\newcommand{\bea}{\begin{eqnarray}}
\newcommand{\eea}{\end{eqnarray}}

\newcommand{\nn}{\nonumber}
\newcommand{\bn}{\begin{enumerate}}
\newcommand{\en}{\end{enumerate}}

\def\IR{\mathbb{R}}

\def\CC{{\cal C}}
\def\CD{{\cal D}}

\def\CF{{\cal F}}

\def\CH{{\cal H}}
\def\CI{{\cal I}}

\def\CL{{\cal L}}
\def\CM{{\cal M}}
\def\CN{{\cal N}}

\def\CR{{\cal R}}
\def\CS{{\cal S}}

\def\a{\alpha}
\def\b{\beta}
\def\g{\gamma}

\def\e{\epsilon}

\def\G{\Gamma}

\def\S{\Sigma}

\def\half{\frac{1}{2}}

\def\Tr{{\rm Tr}}
\def\tr{{\rm Tr}}
\def\det{{\rm det}}

\def\ba{{\mathbf{a}}}
\def\bb{{\mathbf{b}}}
\def\bc{{\mathbf{c}}}
\def\bx{{\mathbf{x}}}
\def\bz{{\mathbf{z}}}
\def\bxi{{\boldsymbol\xi}}
\def\schannel{{\hspace{0.15em}\rangle\hspace{-0.15em}-\hspace{-0.15em}\langle}}

\def\hkquot{\,/\hspace{-0.4ex}/\hspace{-0.4ex}/\hspace{-0.4ex}/\,}
\def\kquot{\,/\hspace{-0.4ex}/\,}
\def\pt{\mathrm{pt}}
\def\tth{\tilde{\theta}}

\theoremstyle{plain}
  \newtheorem{thm}{Theorem}
  \newtheorem{prop}{Proposition}
  \newtheorem{defn}{Definition}
  
  \newtheorem{lem}[prop]{Lemma}
\theoremstyle{remark}

\title{$(0, 4)$ dualities}

\author[a]{Pavel Putrov,}
\author[b]{Jaewon Song,}
\author[a]{and Wenbin Yan}
\affiliation[a]{Walter Burke Institute for Theoretical Physics, California Institute of Technology\\Pasadena, CA 91125, USA}
\affiliation[b]{Department of Physics, University of California, San Diego\\La Jolla, CA 92093, USA}

\emailAdd{putrov@theory.caltech.edu}
\emailAdd{jsong@physics.ucsd.edu}
\emailAdd{wbyan@theory.caltech.edu}
\preprint{CALT-TH 2015-027}

\abstract
{
We study a class of two-dimensional $\CN=(0, 4)$ quiver gauge theories that flow to superconformal field theories. We find dualities for the superconformal field theories similar to the 4d $\CN=2$ theories of class $\CS$, labelled by a Riemann surface $\CC$. The dual descriptions arise from various pair-of-pants decompositions, that involve an analog of the $T_N$ theory. Especially, we find  the superconformal indices of such theories can be written in terms of a topological field theory on $\CC$. We interpret this class of SCFTs as the ones coming from compactifying 6d $\CN=(2, 0)$ theory on $\mathbb{CP}^1 \times \CC$. Moreover, some new dualities of $(0,2)$ and $(2,2)$ theories are also discussed.
}

\begin{document}
\maketitle
\flushbottom

%%%%%%%%%%%%%%%%%%%%%%%%%%%%%%%%%%%%%%%%%%%%%%%%%%%%%
\section{Introduction and summary}

Recent results on two-dimensional gauge theories with $\CN=(0,2)$ theories indicate that the dynamics of such theories can be quite interesting and non-trivial. At the same time the amount of supersymmetry often happens to be sufficient to obtain certain exact results. Such theories have a lot of similarities with $\CN=1$ gauge theories. In particular in \cite{Gadde:2013lxa} it was shown that a large class of $\CN=(0,2)$ theories possess  dualities reminiscent to Seiberg dualities in four dimensions.

In this paper we would like to make a point that $\CN=(0,4)$ theories are likewise similar to $\CN=2$ theories in 4d. In particular we will present  ``2d $\mathcal{N}=(0,4)$ theories of class $\CS$'' analogous to class $\CS$ 4d $\CN=2$ theories introduced in \cite{Gaiotto:2009we, Gaiotto:2009hg}. The latter class of theories has been extensively studied during past years. We show that many statements about $\CN=2$ theories in 4d can be translated into statements about analogous $\CN=(0,4)$ theories. In particular we conjecture dualities among $\CN=(0,4)$ generalized quiver theories analogous to the four-dimensional dualities of \cite{Gaiotto:2009we}.

The main tool that we use to study $\CN=(0,4)$ theories is the superconformal index. We show that it shares a lot of properties with the superconformal index of $\mathcal{N}=2$ 4d theories \cite{Gadde:2009kb,Gadde:2011ik,Gadde:2011uv,Gaiotto:2012xa}. Similarly to the 4d case, the index of ``2d $\mathcal{N}=(0,4)$ theories of class $\CS$'' exhibits a 2d TQFT structure. Following the idea of \cite{Gadde:2010te} we were also able to find an explicit expression for the index of $\CN=(0,4)$ analog of strongly coupled $T_3$ theory with $E_6$ flavor symmetry \cite{Minahan:1996fg}.

Gauge theories with chiral supersymmetry are also interesting because of the possible relation to four-dimensional geometry. Such relation arises from a twisted compactification of a 6d $(2,0)$ SCFT labeled by a Lie algebra $\mathfrak{g}$ on a four-manifold $M_4$. The effective theory in dimension two is usually denoted as $T_\mathfrak{g}[M_4]$. For a 4-manifold of general holonomy one can make a topological twist along $M_4$ such that $T_\mathfrak{g}[M_4]$ has $\mathcal{N}=(0,2)$ supersymmetry. The $(2,0)$ SCFT is a world-volume theory of a stack of M5-branes. Geometrically the twist corresponds to realizing the 4-manifold wrapped by the fivebranes as a coassociative cycle in a 7-dimensional manifold with $G_2$ holonomy embedded into the M-theory space-time. General features of the correspondence $M_4\rightarrow T_\mathfrak{g}[M_4]$ and some particular examples were considered in \cite{Gadde:2013sca,Benini:2013cda}. However identifying $T_\mathfrak{g}[M_4]$ for a generic $M_4$ and $\mathfrak{g}$ is still a very hard task. Therefore considering different concrete examples of 4-manifolds and $\mathfrak{g}$ may help to understand the relation between $M_4$ and $T_\mathfrak{g}[M_4]$ in general.

In the case when 4-manifold $M_4$ is K\"ahler the same twist corresponds to embedding $M_4$ as a complex surface inside a Calabi-Yau threefold. In this case the supersymmetry of the 2d theory $T_\mathfrak{g}[M_4]$ enhances to $\mathcal{N}=(0,4)$. A particular class of such 4-manifolds can be realized by considering holomorphic Lefschetz fibrations, that is holomorphic fibrations of a complex curve with a fixed genus over another curve with possible simple singular fibers. In \cite{Gadde:2014wma} one M5-brane on such 4-manifolds was considered.

One can study even more special class of complex surfaces: products of two complex curves \cite{Benini:2013cda}. In this case it is also possible to consider a twist which preserves $\mathcal{N}=(2,2)$ symmetry in 2d. However the twist preserving $\mathcal{N}=(0,4)$ is more interesting in a way, because in this case the product of curves can be understood just as a particular choice of $M_4$. We would like to conjecture that ``class $\CS$ 2d $\mathcal{N}=(0,4)$ gauge theories'' that we consider in the paper can be realized as $T_{\mathfrak{g}}[\mathbb{CP}^1 \times \CC]$ where $\CC$ is a Riemann surface with possible punctures. In this way the relation to $\mathcal{N}=2$ 4d theories of class $\CS$ becomes transparent. The dualities among different 2d theories from class $\CS$ then can be understood as corresponding to different decompositions of $\CC$ into pairs of pants. From this conjecture it also follows that the 2d TQFT describing the index is a reduction of Vafa-Witten 4d TQFT \cite{Vafa:1994tf} on $\mathbb{CP}^1$. This relation may shed a light on better understanding of Vafa-Witten (VW) TQFT from categorical point of view, i.e. as functor from the category of 3-cobordisms to the category of vector spaces. So far in most of the literature the VW partition function was studied on a particular, usually closed 4-manifold. Some of the progress in understanding of VW TQFT as a functor was made in \cite{Gadde:2013sca}, where the gluing procedure of certain 4-manifolds was considered.

This interpretation is in agreement with recent calculations of the $S^2\times T^2$ index of general $\mathcal{N}=1$ 4d gauge theories \cite{Honda:2015yha,Benini:2015noa} with topological twist along $S^2$. The result has an expression that can be interpreted as the index of a $(0,2)$ 2d theory. In particular, in the case when $\CC$ is a three-punctured sphere and $\mathfrak{g}=\mathfrak{su}(3)$, by solving an integral equation we find index which agrees with the result from \cite{Gadde:2015xta}. In that paper the authors propose a $\mathcal{N}=1$ 4d gauge theory that flows in the IR to a strongly coupled 4d $\mathcal{N}=2$ $T_3$ theory with $E_6$ flavor symmetry and calculate its $S^2\times T^2$ twisted index.

However the aim of this paper is not  to focus on the 4-manifold realization of two-dimensional theories or their 4d gauge theory origin, but to study them purely from two-dimensional point of view. The relation to 4-manifolds will be explored in detail elsewhere. Let us note that currently there are almost no non-trivial results about gauge theories with $\mathcal{N}=(0,4)$ supersymmetry in the literature. Our work can be considered as a step towards improving this situation.

This paper is organized as follows. In section 2 we introduce $\mathcal{N}=(0,4)$ (and $(4,4)$) class $\CS$ theories with gauge group being a product of several copies of $SU(2)$ and study their properties. In section 3 we consider generalization to $SU(N)$. In section 4 we show that $\mathcal{N}=(0,2)$ (and  $\mathcal{N}=(2,2)$) SQCDs with $SU(N)$ gauge group and $2N$ flavors share certain similarities.

%%%%%%%%%%%%%%%%%%%%%%%%%%%%%%%%%%%%%%%%%%%%%%%%%%%%%%%%
\section{Dualities of $SU(2)$ generalized quiver}

\subsection{$SU(2)$ with 4 flavors and its crossing symmetry}
\label{section:SU2basic}

Let us consider the simplest possible two-dimensional SQCD with $\mathcal{N}=(0,4)$ supersymmetry and $SU(2)$ gauge group. Such a theory contains at least $(0,4)$ vector multiplet $(U,\Theta)$ consisting of a $(0,2)$ Vector multiplet $U$ and $(0,2)$ Fermi multiplet in adjoint representation (see appendix \ref{appendix04} for a brief review of 2d $(0,2)$ and $(0,4)$ theories). The vector multiplet contributes in total $-4$ to the 't Hooft anomaly coefficient\footnote{In appendix \ref{appendixanomalies} we define its normalization and give a basic review of 't Hooft anomalies in 2d} of $SU(2)$ gauge group. If we want to add matter fields in the fundamental representation, the minimal choice that cancels the gauge anomaly from the vector multiplet is four fundamental $(0,4)$ hypermultiplets $(\Phi,\tilde{\Phi})$. In order for the theory to be $(0,4)$ supersymmetric we also have to choose the following superpotential:
\begin{equation}
 W=\tilde{\Phi}\Theta\Phi\,.
\end{equation}

The constructed theory has $SU(4)$ flavor symmetry as well $U(1)_B$ baryonic global symmetry. The hypermultiplets form the following representation\footnote{We follow the notations of \cite{Slansky:1981yr} for group representations throughout the paper.} w.r.t. $SU(2)\times SU(4)\times U(1)_B$:

\begin{equation}
 \mathbf{(2,4)_{+1}+(\bar{2},\bar{4})_{-1}}\,.
\end{equation}

As we will show later in the paper, this theory shares a lot of properties with the analogous 4d $\mathcal{N}=2$ theory, which was studied in great detail already in \cite{Seiberg:1994aj}. In particular, the flavor symmetry is enhanced to $SO(8)$ at the classical level. This can be easily seen from the fact that for $SU(2)$ we have $\mathbf{2=\bar{2}}$ and $\mathbf{4_{+1}+\bar{4}_{-1}=8_v}$ of $SO(8)\supset SU(4)\times U(1)$. Since the $(0,4)$ vector multiplet does not have any scalar fields, the theory has no Coulomb branch. The Higgs branch is defined by the triplet of $D$-term conditions and can be represented as the $\mathbb{H}^{8}\hkquot SU(2)$ hyper-K\"ahler quotient. It is the same as the Higgs branch of 4d $\CN=2$ theory and does not acquire any quantum corrections. The scalar fields of $(\Phi,\tilde\Phi)$ transform in representation $(2,1)$ of $SU(2)_{R}^-\times SU(2)_{R}^+$ of UV R-symmetry group. Following the arguments of \cite{Witten:1997yu} one then expects $SU(2)_{R}^+$, under which the scalars parametrizing the Higgs branch transform trivially, to be the $SU(2)_R$ R-symmetry of the small $\mathcal{N}=4$ superconformal algebra (SCA) in the right-moving sector of the IR SCFT.

The hyper-K\"ahler dimension of the Higgs branch is $8-3=5$ which is the same as twice the 't Hooft anomaly coefficient of $SU(2)_{R}^+$ or, equivalently, the level of the affine $\widehat{SU(2)}$ R-symmetry algebra in the IR SCFT. It follows that the central charges of the theory are
\begin{equation}
 c_R=6\cdot 5=30,\qquad c_L=20
\end{equation}
where we also used the fact that $c_L-c_R$ equals to the gravitational anomaly which is easily calculated in the UV as the difference between the numbers of left and right moving complex fermions.

We would like to conjecture that the spectrum of the (0,4) SCFT at the IR fixed point is also invariant under the action of $SO(8)$ triality which permutes vector representation $\mathbf{8_v}$ and two spinor representations $\mathbf{8_s}$ and $\mathbf{8_c}$. Unlike in the $\mathcal{N}=2$ 4d case, we do not need to accompany the triality action with a transformation of the the gauge coupling because it is not marginal in 2d. There are also no other apparent exactly marginal deformations of the (0,4) $SU(2)$ gauge theory in the UV, since there is no FI parameter for $SU(2)$ gauge group and the superpotential is completely fixed by $(0,4)$ supersymmetry.

As in the 4d $\mathcal{N}=2$ case \cite{Gaiotto:2009we}, the symmetry under triality can be reformulated in a different way, which will be useful later in the paper when we consider more general quiver theories. Let us define 2d $\mathcal{N}=(0,4)$ theory $T_2^{(0,4)}$ analogous to 4d $\mathcal{N}=2$ theory $T_2$ as the theory of free $(0,2)$ chiral multiplets (``half-hypers'') in the tri-fundamental representation $\mathbf{(2,2,2)}$ of $SU(2)^3$ flavor symmetry. In quiver notation we will depict this theory as a triangle with 3 external legs corresponding to $SU(2)$ flavor groups (see Fig. \ref{fig:SU2-quiver-notations}a). As usual, we will represent $SU(N)$ vector multiplet as a circle (see Fig. \ref{fig:SU2-quiver-notations}b). Then the $(0,4)$ $SU(2)$ gauge theory with 4 flavors can be represented as two copies of $T_2^{(0,4)}$ glued together by a $SU(2)$ vector multiplet gauging the diagonal subgroup of $SU(2)\times SU(2)$ (see Fig. \ref{fig:SU2-schannel}).

\begin{figure}[ht]
\centering
\includegraphics[scale=1]{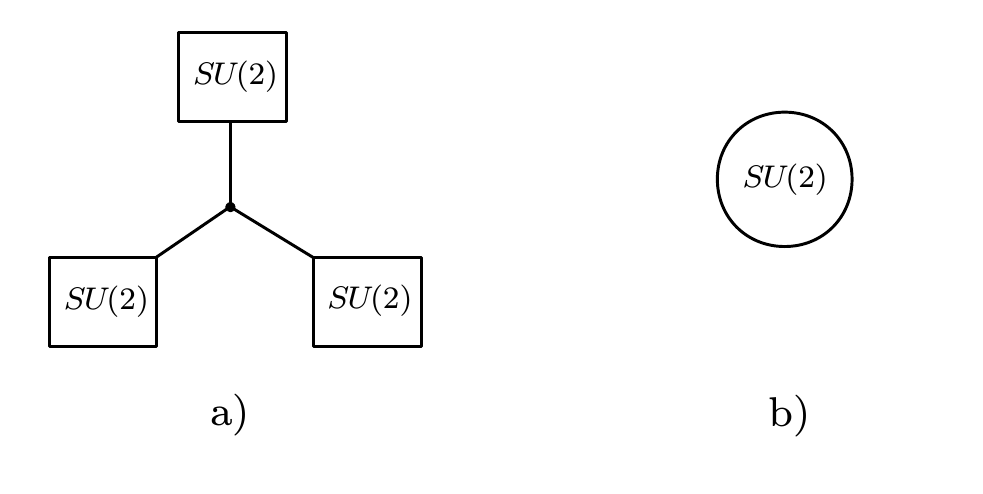}
\caption{The quiver notations for: a) theory $T_2^{(0,4)}$ of 8 chiral multiplets in tri-fundamental representation of $SU(2)^3$ flavor symmetry, b) $(0,4)$ $SU(2)$ vector multiplet.}
\label{fig:SU2-quiver-notations}
\end{figure}

\begin{figure}[ht]
\centering
\includegraphics[scale=1]{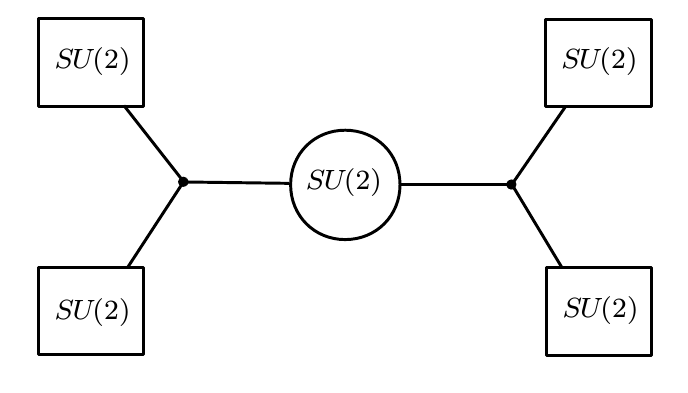}
\caption{The quiver notation for the theory obtained by gauging the diagonal subgroup of two $SU(2)$ flavor symmetries from two different copies of $T_2^{(0,4)}$ with $(0,4)$ $SU(2)$ vector multiplet.}
\label{fig:SU2-schannel}
\end{figure}

The flavor symmetry of the resulting theory is $SU(2)^4$ which is enhanced to $SO(8)$. The chiral fields in the hypermultiplets form the following representation of the flavor group:
\begin{equation}
 \mathbf{8_v=(2,2,1,1)+(1,1,2,2)}\,.
\end{equation}
Two spinor representations of $SO(8)$ decompose as:
\begin{equation}
\begin{array}{c}
\mathbf{8_s=(1,2,1,2)+(2,1,2,1)}\,, \\
\mathbf{8_c=(1,2,2,1)+(2,1,1,2)}\,.
\end{array}
\end{equation}
Therefore the invariance of the spectrum under $SO(8)$ triality is equivalent to the symmetry under permutations of $SU(2)$ factors in $SU(2)^4$ flavor symmetry, or crossing symmetry of the quiver diagram (see Fig. \ref{fig:SU2-crossing}).

\begin{figure}[ht]
\centering
\includegraphics[scale=0.75]{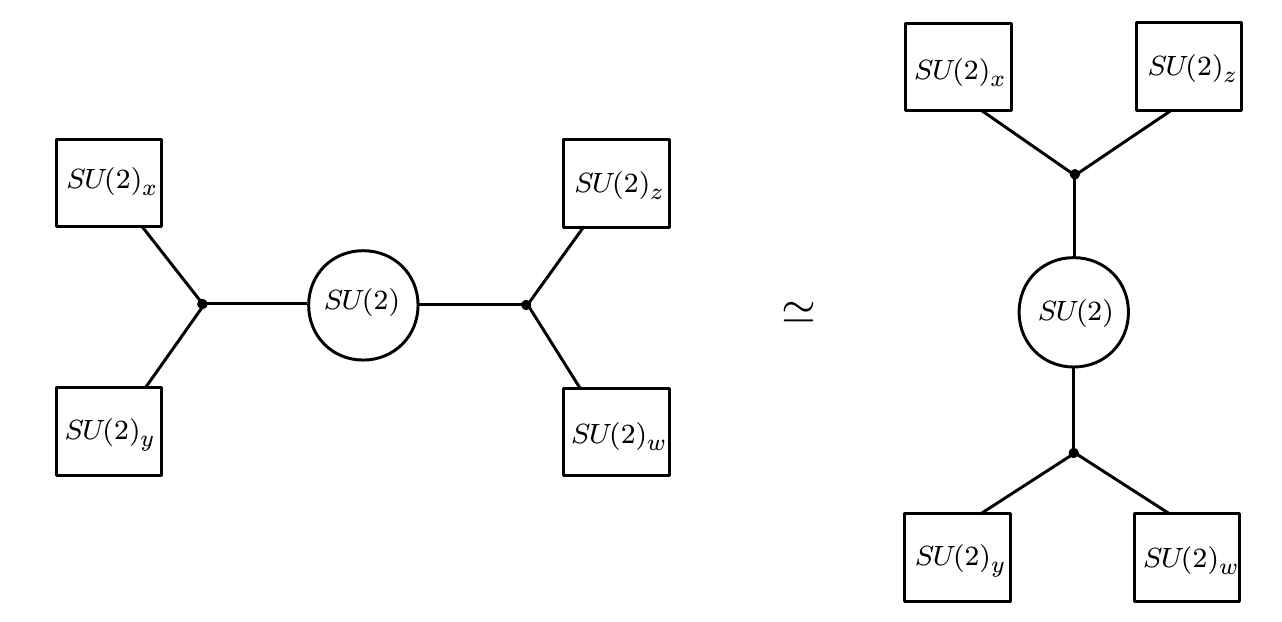}
\caption{The symmetry under exchange of $SU(2)$ factors in the flavor symmetry of the theory can be interpreted as the crossing symmetry of the quiver diagram. The letters $x,y,z,w$ used to distinguish various $SU(2)$ factors and later in the text denote the corresponding $SU(2)$ flavor fugacities in the elliptic genus.}
\label{fig:SU2-crossing}
\end{figure}

The statement can be checked by calculating the 2d superconformal index (also known as flavored elliptic genus\footnote{In this paper we are using ``superconformal index'' and ``elliptic genus'' interchangeably.}) of the theory \cite{Gadde:2013wq,Benini:2013nda,Gadde:2013dda,Benini:2013xpa}. The NS-NS index of the theory at hand can be calculated as the following integral (see appendix \ref{appendixindex} for a review of the superconformal index in 2d):
\begin{equation}
 \CI^{(0,4)}_{\schannel}(x,y,z,w;v;q)=\half\int\limits_\text{JK} \frac{d\xi}{2\pi i \xi}\, \CI^{(0,4)}_{T_2}(x,y,\xi;v;q)\,\CI^{(0,4)}_{V,SU(2)}(\xi;v,q)\,\CI^{(0,4)}_{T_2}(1/\xi,z,w;v,q) \ ,
 \label{I04SO8int}
\end{equation}
taken over a certain contour ``JK'' which corresponds to taking a sum of Jeffrey-Kirwan residues. For example, in the case of rank one gauge group the contour encircles only the poles coming from scalar fields with positive (or, equivalently, negative) charges w.r.t. the Cartan $U(1)$. The factors entering the integrand are
\begin{equation}
 \CI^{(0,4)}_{T_2}(x,y,z;v;q)\equiv \frac{1}{\theta(v\,x^\pm y^\pm z^\pm)}\,,
 \label{index-T2}
\end{equation}
the index of $T_2^{(0,4)}$ (tri-fundamental half-hyper) where $x, y$ and $z$ denote the fugacities corresponding to $SU(2)^3$ flavor symmetries, and
\begin{equation}
 \CI^{(0,4)}_{V,SU(2)}(\xi;v,q)\equiv (q;q)^2\theta(q/v^2){\theta(q\,\xi^{\pm 2}/v^{2})\theta(\xi^{\pm 2})}\,,
 \label{index-V2}
\end{equation}
the index of $(0,4)$ $SU(2)$ vector multiplet. Here and throughout the paper we use the common notation:
\begin{equation}
 f(x^\pm)\equiv f(x)f(x^{-1}).
\end{equation}
The fugacity $v$ corresponds to $U(1)_v$ global symmetry -- anti-diagonal Cartan of $SU(2)_{R}^-\times SU(2)_{R}^+$ R-symmetry which commutes with the supercharges used to calculate the index. The index can be understood as the $(0,2)$ index where the IR $U(1)_R$ R-symmetry is chosen as the Cartan of $SU(2)_{R}^+$ and $U(1)_v$ plays the role of a flavor symmetry. See appendix \ref{appendixindex} for details. Since the theory has only the Higgs branch, we expect the elliptic genus to coincide with geometrically defined $(0,2)$ equivariant elliptic genus \cite{Kawai:1994np} of the Higgs branch manifold $X=\mathbb{H}^{8}\hkquot SU(2)$ with empty vector bundle of left-moving fermions:
\begin{equation}
 \CI^{(0,4)}_{\schannel}=\int_X\det\frac{F_T}{\theta(e^{F_T})}.
 \label{EG-nobundle}
\end{equation}
where $F_T$ is the curvature on the tangent bundle $TX$.

The integral (\ref{I04SO8int}) can be calculated explicitly by residues. The result contains 8 terms, each of which has the form of ratio of products of theta functions. To make the formula simpler let us denote the collection of $SU(2)^4$ fugacities $(x,y,z,w)$ as $\bx$ which can be understood as the element of the maximal torus of $SO(8)$. In the limit $q\rightarrow 0$ the index becomes the same as Hilbert series of $X$ calculated in \cite{Benvenuti:2010pq,Hanany:2010qu}, which can be written as
\be
 \CI^{(0, 4)}_\schannel(\bx; v; q\to0) = \sum_{k=0}^{\infty} \chi^{SO(8)}_{k \theta }(\bx) v^{2k} = \mathbf{1}+\mathbf{28}\,v^2+\mathbf{300}\,v^4+\mathbf{1925}\,v^6+\ldots \ ,
\ee
where $\theta$ denotes the highest root of $SO(8)$, and $\chi_{k \theta}$ is the character for the Dynkin label given by $k \theta$. For the sake of simplicity we later denote characters by the dimension of the corresponding representations. When $k=1$, this is the character of the adjoint representation.
This is the same as the Hilbert series of the (centered) one $SO(8)$ instanton moduli space \cite{Garfinkle,VinbergPopov}, where the first equality also holds for arbitrary simple gauge group $G$. The Hilbert series of (centered) 1-instanton moduli space can also be written as a sum over root vectors \cite{Keller:2011ek,Keller:2012da} as
\be
\textrm{HS}_{G}(\mu, \phi) =  \sum_{\g \in \Delta_l} \frac{e^{(h^\vee - 1)\g \cdot \phi /2}}{(1-e^{\mu+\g \cdot \phi})(e^{\g \cdot \phi/2} - e^{-\g \cdot \phi /2})\prod_{\g^\vee \cdot \a =1} (e^{\a \cdot \phi/2} - e^{-\a \cdot \phi/2})} \ ,
\ee
where $h^\vee$ is the dual Coxeter number of $G$, and $\Delta_l$ is the set of long roots and $\phi$ is an element in the Cartan. We identify $v = e^{\mu/2}$, $\bx = e^{\phi}$. There are poles at $v^2 \bx^\g = 1$ for $\g \in \Delta_l$.

One can show that the index has the following structure:
\begin{equation}
 \CI^{(0,4)}_{\schannel}(\bx;v;q)=\frac{\tilde{\CI}^{(0,4)}_{\schannel}(\bx;v;q)}{\prod\limits_{\lambda \in\mathbf{28}}\theta(v^2\,\bx^\lambda)} \ ,
 \label{I04SO8}
\end{equation}
where $\mathbf{28}$ denotes adjoint representation of $SO(8)$ and the function $\tilde{\CI}^{(0,4)}_{\schannel}(\bx;v,q)$ is regular in $\bx$. The denominator of (\ref{I04SO8}) can be understood as the contribution of gauge invariant mesons constructed from bilinear combinations of the chiral fields because
\begin{equation}
 \mathrm{Sym}^2\,\mathbf{(2,8_v)}=\mathbf{(1,28)}+\mathbf{(3,1+35_v)} \ ,
\end{equation}
where two numbers in each pair denote the representations w.r.t. $SU(2)$ gauge and $SO(8)$ flavor group respectively. The complex dimension of the Higgs branch is 10 and the numerator of (\ref{I04SO8}) formally corresponds to additional conditions on these 28 mesons from D-term constraints (cf. \cite{Benvenuti:2010pq,Hanany:2010qu}).

The index has the following expansion w.r.t. $q$ and $v$ written in terms of $SO(8)$ characters:
\be
\begin{split}
 {\CI}^{(0,4)}_{\schannel}(\bx;v;q)= &\left(\mathbf{1}+\mathbf{28}\,v^2+\mathbf{300}\,v^4+\mathbf{1925}\,v^6+\ldots\right)\\
& \qquad +\left(\mathbf{(1+28)}+(2\cdot\mathbf{28}+\mathbf{300+350})v^2+\ldots\right)\,q
 +\ldots
 \label{I04SO8series}
\end{split}
\ee
One can see that only $SO(8)$ triality invariant representations appear in the index.

The crossing symmetry of the index (\ref{I04SO8}) can be proven explicitly, not just term by term in $q$ and $v$ expansion. To do this let us consider the difference between indices that differ by a non-trivial transposition of two $SU(2)$ flavor fugacities:
\begin{equation}
 \CI_\Delta^{(0,4)}(x,y,z,w;v;q)\equiv {\CI}^{(0,4)}_{\schannel}(x,y,z,w;v;q)-{\CI}^{(0,4)}_{\schannel}(x,z,y,w;v;q).
 \label{I04Delta}
\end{equation}
Using the explicit expression for the index it is easy to show that $\CI_\Delta^{(0,4)}(x,y,z,w;v;q)$ has no poles in variables $(x,y,z,w)$ (i.e. the residues from two terms in (\ref{I04Delta}) cancel each other). The theory has anomaly coefficient $2$ w.r.t. each $SU(2)$ flavor symmetry factor. Therefore if we further define
\begin{equation}
 \tilde{\CI}_\Delta^{(0,4)}(x,y,z,w;v;q)\equiv \CI_\Delta^{(0,4)}(x,y,z,w;v;q)\cdot\left(\theta(x^\pm)\theta(y^\pm)\theta(z^\pm)\theta(w^\pm)\right)^4
\end{equation}
it will be a function elliptic in $(x,y,z,w)$ (i.e. invariant under the shifts $x\rightarrow qx$, $y\rightarrow qy$, etc.) and with no poles. It follows that $\tilde{\CI}_\Delta^{(0,4)}(x,y,z,w;v;q)$ should be constant in $x,y,z,w$. And since $\CI_\Delta^{(0,4)}(x,y,z,w;v;q)$ has no pole at $x=1$ this constant should be zero. This proves the crossing symmetry property of the index ${\CI}^{(0,4)}_{\schannel}(x,y,z,w;v;q)$, namely:
\begin{equation}
  {\CI}^{(0,4)}_{\schannel}(x,y,z,w;v;q)-{\CI}^{(0,4)}_{\schannel}(x,z,y,w;v;q)=0\,.
  \label{index-crossing}
\end{equation}

The triality outer-automorphism of $SO(8)$ can be understood as the Weyl group action of $F_4$ if we embed $SO(8)\subset F_4$. This means that the series (\ref{I04SO8series}) can be formally rewritten in terms of characters of $F_4$ representations:
\be
\begin{split}
 {\CI}^{(0,4)}_{\schannel}(\bx;v;q) =&\left(\mathbf{1}+(\mathbf{52-26}+2\cdot\mathbf{1})\,v^2+\mathbf{300}\,v^4+\ldots\right)\\
 & \qquad \qquad +\left((\mathbf{52-26}+3\cdot\mathbf{1})+\ldots\right) q
 +\ldots
\end{split}
\ee
The index of the analogous $\mathcal{N}=2$ 4d theory has similar property \cite{Gadde:2009kb}. As in the 4d case, it does not follow that the global symmetry actually enhances from $SO(8)$ to $F_4$ in the IR SCFT because there is no conserved current of $F_4$.

\subsection{Dualities of quiver theories and the TQFT structure of the index} \label{subsec:SU2TQFT}

\subsubsection{Elliptic genus and 2d TQFT} \label{section-TQFT}

Similarly to the 4d $\mathcal{N}=2$ case \cite{Gadde:2009kb}, the crossing symmetry of the index (\ref{I04SO8int}) indicates that (\ref{index-T2}) and (\ref{index-V2}) can be used to define a 2d TQFT. Namely, let us define the Hilbert space of the 2d TQFT associated to a circle as the following space of meromorphic functions\footnote{This space can be understood as the space of meromorphic sections of $\CL^{-4}\rightarrow \CM_\text{flat}(T^2_\tau,SU(2))$, see appendix \ref{appendixanomalies} for details. It would be interesting to check explicitly if this is the Hilbert space of VW TQFT associated to $\mathbb{CP}^1\times S^1$, or, equivalently, the BPS sector of the Hilbert space of $T_{\mathfrak{su}(2)}[\mathbb{CP}^1\times S^1]$ quantized on $T^2_\tau$.}:
\begin{equation}
 \CH_{S^1}^{(0,4)}=\{ f:\mathbb{C}^* \rightarrow \mathbb{C}\,|\, f(x)=f(1/x),\,f(qx)=q^{4}x^{8}f(x)\}\,.
 \label{TQFT-Hilbert-space}
\end{equation}
Then define the basic building blocks of 2d TQFT:

\begin{equation}
 \begin{tabular}{|c|c|}
 \hline
$\vcenter{\hbox{\includegraphics[scale=0.7]{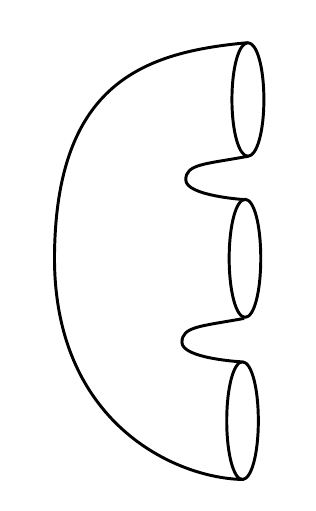}}}$ &
$
\begin{array}{rrcl}
  C:&\, \mathbb{C}& \longrightarrow & \CH_{S^1}^{(0,4)}\otimes \CH_{S^1}^{(0,4)} \otimes \CH_{S^1}^{(0,4)}\\
  \\
 & 1 & \longmapsto&\CI^{(0,4)}_{T_2}(x,y,z;v;q)
 \end{array}
$
\\
\hline
$\vcenter{\hbox{\includegraphics[scale=0.7]{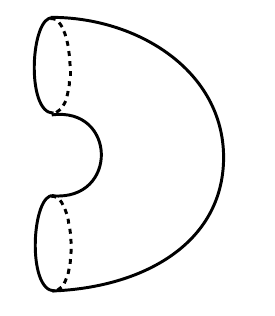}}}$ &
$
\begin{array}{rrcl}
  \eta:&\,\CH_{S^1}^{(0,4)}\otimes \CH_{S^1}^{(0,4)}&\longrightarrow & \mathbb{C}\\
  \\
 & f(x,y) & \longmapsto& \half\int\limits_\text{JK}\frac{d\xi}{2 \pi i \xi}\,\CI^{(0,4)}_{V,SU(2)}(\xi;v;q)f(\xi,\xi)
 \end{array}
$\\
\hline
\end{tabular}
\label{TQFT-blocks}
\end{equation}
Note that the last property in  (\ref{TQFT-Hilbert-space}) is required for the integrand in the definition of $\eta$  to be elliptic. Using $\eta$ and $C$ one can define a commutative product $\mu$ on $ \CH_{S^1}^{(0,4)}$:
\begin{equation}
 \begin{tabular}{|c|c|}
 \hline
$\vcenter{\hbox{\includegraphics[scale=0.5]{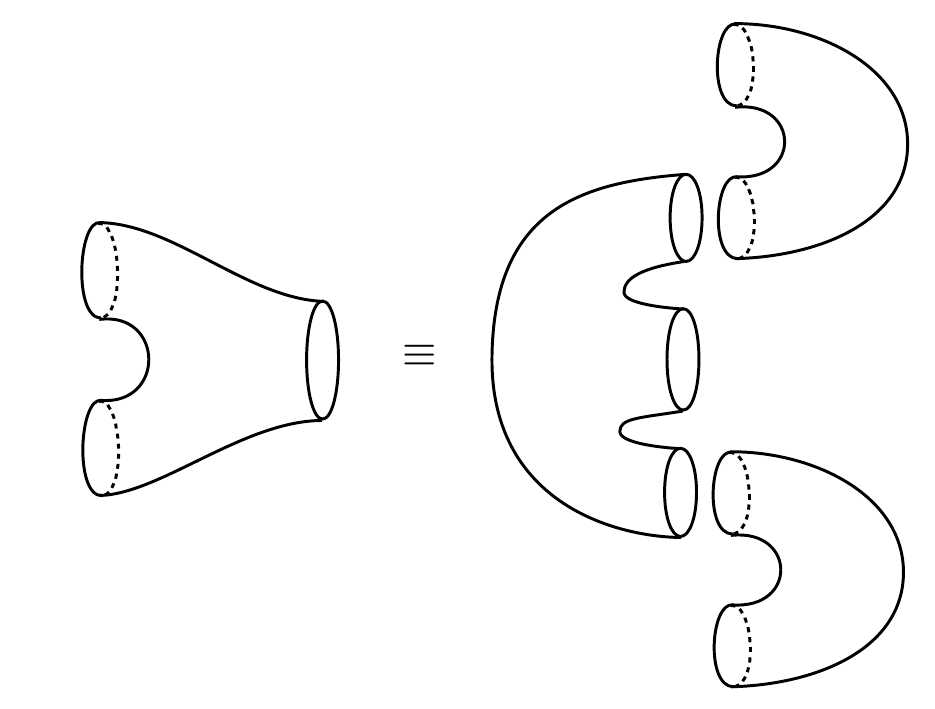}}}$ &
$
\begin{array}{c}
\mu\equiv (\eta\otimes\mathrm{id}\otimes \eta)\circ(\mathrm{id}\otimes C \otimes \mathrm{id})\\
\\
\mu:\,  \CH_{S^1}^{(0,4)}\otimes \CH_{S^1}^{(0,4)} \longrightarrow \CH_{S^1}^{(0,4)}
\end{array}
  $
\\
\hline
\end{tabular}
\end{equation}
where $\mathrm{id}:\CH_{S^1}^{(0,4)} \longrightarrow \CH_{S^1}^{(0,4)}$ is the identity map. The crossing symmetry property (\ref{index-crossing}) of the index is then equivalent to the associativity of $\mu$ which can be formulated in the following way:
\begin{equation}
 \begin{tabular}{|c|c|}
 \hline
$\vcenter{\hbox{\includegraphics[scale=0.5]{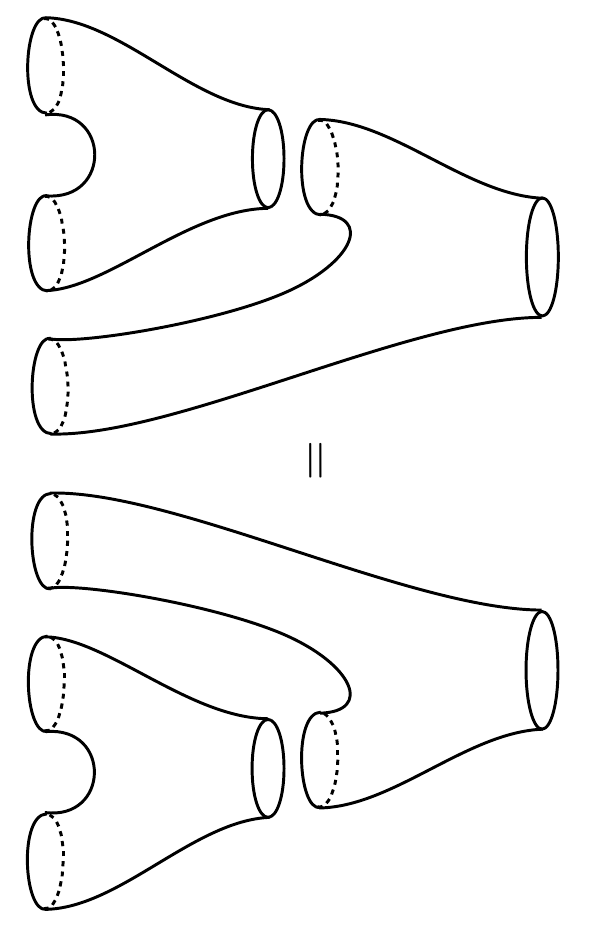}}}$ &
$
\begin{array}{c}
\mu \circ (\mu\otimes\mathrm{id})\\
\text{\rotatebox{90}{$=$}}
\\
\mu \circ (\mathrm{id}\otimes\mu)
\end{array}
  $
\\
\hline
\end{tabular}
\end{equation}

\subsubsection{Dualities between generalized quiver theories}

As in \cite{Gaiotto:2009we}, the crossing symmetry property of the IR spectrum of the theory depicted in Fig. \ref{fig:SU2-schannel} can be used to deduce IR dualities between various theories constructed from the basic building blocks in Fig. \ref{fig:SU2-quiver-notations}.

For example, consider a theory defined by the quiver in the l.h.s. of Fig. \ref{fig:tree-6flavor-duality}. Applying the crossing symmetry transformation in Fig. \ref{fig:SU2-crossing} to the middle part we get a different theory corresponding to the quiver in the r.h.s. of Fig. \ref{fig:tree-6flavor-duality}. From the point of view of 2d TQFT defined above the index of the theory is the partition function (which can be understood as an element of $\in (\CH_{S^1}^{(0,4)})^{\otimes 6}$) of the sphere with 6 punctures.
\begin{figure}[ht]
\centering
\includegraphics[scale=1]{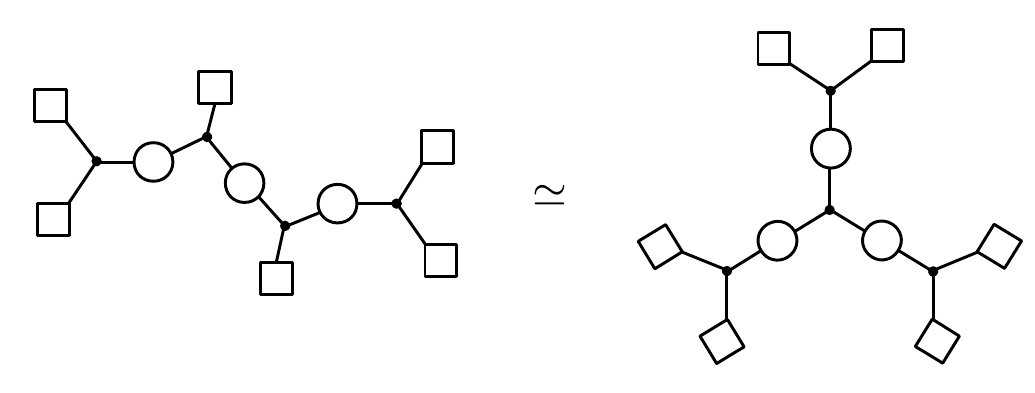}
\caption{Duality between two different $(0,4)$ theories with $SU(2)^3$ gauge group and $SU(2)^6$ flavor symmetry. For the sake of simplicity we suppress $SU(2)$ inscribed inside squares and circles of the quivers.}
\label{fig:tree-6flavor-duality}
\end{figure}
The first theory is a linear quiver gauge theory, and the second one contains trifundamental hypermultiplet coupled to three $SU(2)$ gauge groups.

One can consider another example of duality between two distinct 2d (0,4) theories that follows from the crossing symmetry as depicted in Fig. \ref{fig:genus2-duality}. The index of such theory can be understood as the 2d TQFT partition function of a genus two Riemann surface.

\begin{figure}[ht]
\centering
\includegraphics[scale=1]{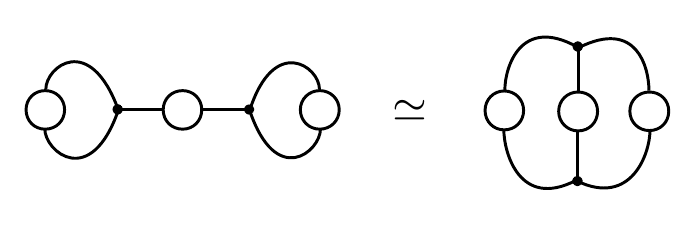}
\caption{Duality between two different $(0,4)$ theories with $SU(2)^3$ gauge group.}
\label{fig:genus2-duality}
\end{figure}

However, in the case when quiver has loops the physics is a little more complicated because the gauge group is not completely broken. Consider a theory corrsponding to a quiver with $g$ loops and $n$ external legs. In terms of 2d TQFT the index is the partition function of a genus $g$ Riemann surface with $n$ punctures $\CC_{g,n}$. The theory has $3g-3+n$ copies of $SU(2)$ vector multiplet and $2g-2+n$ copies of $SU(2)$ trifundamental chiral multiplet $T_2^{(0,4)}$. The resulting theory has $SU(2)^n$ flavor symmetry. When $g>0$ a $U(1)^g$ part of the gauge symmetry remains unbroken for general expectation values of hyper-multiplets. Each unbroken $U(1)$ factor is the the diagonal maximal torus of the gauge group $\prod_{i \in \textrm{loop}} SU(2)_i$ associated to the loop in the quiver. Following the authors of \cite{Hanany:2010qu} in this case we will refer to the moduli space $X$ parametrized by massless gauge-invariant combinations of hypermultiplets as Kibble branch. The naive counting of its dimensions -- as $n_h-n_v$ where $n_{h,v}$ are the numbers of hyper- and vector multiplets of the theory respectively -- does not work in this case. The reason is that $SU(2)^{3g-3+n}$ does not act freely on $\mathbb{H}^{4(2g-2+n)}$ space of hyper-multiplets. The mismatch of the quaternionic dimension is given by $g$, the rank of the unbroken part of the gauge group. It follows that the Kibble branch CFT should have the following central charges:
\begin{equation}
 c_R=6(n_h-n_v+g)=6(n+1),\qquad c_L=4(n+1)+2g,
\end{equation}
where we calculated $c_L$ from the gravitational anomaly. Let us note that $c_L>2c_R/3$ when $g>0$. This is beacuse, unlike in the case when quiver has no loops, unbroken directions of the gauge group give rise to a non-empty complex rank $2g$ bundle $E$ of left-moving Fermions, the only remnant of the usual Coulomb branch that would appear for $(4,4)$ theories. Again, as for the basic theory in section \ref{section:SU2basic}, at least for the large values of scalar fields, we expect the IR SCFT to have a sigma-model description in terms of target space $X\cong \mathbb{H}^{4(2g-2+n)}\hkquot SU(2)^{3g-3+n}$, where $(0,2)$ chiral multiplets play the role of complex coordinates, and a holomorphic vector bundle\footnote{In general the dimension of the fiber (i.e. the number of massless left-moving fermions) can depend on a point in the moduli space, $E$ then should be considered as a sheaf.} of (0,2) Fermi multiplets $E\rightarrow X$. The index then has the meaning of the following equivariant characteristic class \cite{Kawai:1994np}:
\begin{equation}
 \CI^{(0,4)}=\int_X\det\frac{F_T}{\theta(e^{F_T})}\cdot \det\,\theta(e^{F_E})
\end{equation}
where $F_E$ and $F_T$ are the curvatures on $E$ and $TX$ respectively. In the next section we consider example with $g=1$ and $n=1$ in detail.

Let us note that the relation $c_R=6\cdot(2k_{SU(2)_R^+})\equiv 6(n_h-n_v)$ between the right-moving central charge and the anomaly of $SU(2)_R^+$ UV R-symmetry does not work when $g>0$ for the following reason. In the sigma-model description $SU(2)_R^+$ now acts not only on the right-moving fermions living in the tangent bundle of the Kibble branch, but also on the left-moving fermions in the complex rank $2g$ vector bundle $E$. Therefore, similarly to what happens on the Coulomb branch of $(4,4)$ theories \cite{Witten:1997yu}, we expect that in IR SCFT $SU(2)_R^+$ splits into two symmetries, one is left-moving global symmetry $SU(2)$ affine symmetry with level $g$, and the other is right-moving $SU(2)$ affine R-symmetry with level $(n_h-n_v+g)$, which is in agreement with the value of $c_R$. In the UV we only see the diagonal of these two symmetries, $SU(2)_R^+$, with anomaly coefficient being half the difference of affine algebras levels, $(n_h-n_v)/2$.

\subsubsection{Duality to a Landau-Ginzburg model}

\begin{figure}[ht]
\centering
\includegraphics[scale=1]{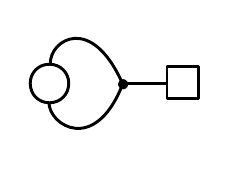}
\caption{The quiver of $(0,4)$ theory with $SU(2)$ vector multiplet $(U,\Theta)$ and a hyper multiplet $(\Phi,\tilde{\Phi})$ in adjoint representation.}
\label{fig:torus-1puncture}
\end{figure}

Consider the theory associated to the quiver in Fig. \ref{fig:torus-1puncture}.
One can show that the index of this theory satisfies the following identity:

\begin{equation}
\begin{split}
 \half\int\limits_\text{JK}\frac{d\xi}{2\pi i \xi}\,\CI^{(0,4)}_{V,SU(2)}(\xi;v;q)\,&\CI^{(0,4)}_{T_2}(\xi,\xi,x;v;q)=\\
&= \frac{1}{\theta(v/x)\theta(vx)}\cdot \frac{\theta(q/v^4)}{\theta(v^2)\theta(v^2/x^2)\theta(v^2x^2)}\,\cdot \theta(v/x)\theta(vx)
\label{genus-one-LG}
\end{split}
\end{equation}
where we explicitly factored out the contribution from decoupled chiral fields $(\Tr\,\Phi,\Tr\,\tilde\Phi)$ spanning $\mathbb{C}^2$. The second factor in right hand side can be understood as the index of the $(0,2)$ Landau-Ginzburg model with three $(0,2)$ chiral multiplets $\Phi_{1,2,3}$, one Fermi multiplet $\Gamma$ and the superpotential
\begin{equation}
 W=\Gamma (\Phi_1\Phi_2-\Phi_3^2)\,.
 \label{LG-potential}
\end{equation}
The superpotential (\ref{LG-potential}) implies the condition
\begin{equation}
 \Phi_1\Phi_2-\Phi_3^2=0 \label{conifold-condition}
\end{equation}
which is the equation describing an embedding of $\mathbb{C}^2/\mathbb{Z}_2$ into $\mathbb{C}^3$. The chiral fields $\Phi_i$ can be mapped to the following gauge invariant operators in the chiral ring of the original gauge theory:
\begin{equation}
 \begin{array}{rl}
          \Phi_1 & =\Tr\,\Phi^2\,,\\
          \Phi_2 & =\Tr\,\tilde\Phi^2\,,\\
          \Phi_3 & =\Tr\,\Phi\tilde{\Phi}\,.
         \end{array}
\end{equation}
Then the condition (\ref{conifold-condition}) follows from the condition $[\Phi,\tilde\Phi]=0$ imposed by the superpotential associated to $\Theta$.

The first two factors in the right hand side of (\ref{genus-one-LG}) describe $(0,2)$ chiral fields spanning the Kibble branch of the theory, $X=\mathbb{C}^8\hkquot SU(2)\cong \mathbb{C}^2\times \mathbb{C}^2/\mathbb{Z}_2$, and in the limit $q\rightarrow 0$ they reproduce its Hilbert series \cite{Hanany:2010qu}. The last factor in $(\ref{genus-one-LG})$ is the contribution of a complex rank two holomorphic vector bundle $E\rightarrow X$ of left-moving fermions. It appears in this case because the gauge group is not completely broken (contrary to the case when a quiver does not have any loops, the gauge group is completely broken and $E$ is empty). In terms of the original gauge theory the fibers of the bundle $E$ are generated by massless gauge invariant Fermi multiplets $\Tr\Lambda \Phi$ and $\Tr\Lambda \tilde\Phi$, where is $\Lambda$ is the $(0,2)$ field strength Fermi multiplet constructed from the vector multiplet $U$. From the dimensions of the target space and the bundle $E$ we conclude that
\begin{equation}
 \begin{array}{rl}
          c_R & =12\,,\\
          c_L & =10.
         \end{array}
\end{equation}

Let us note that in this particular case ($g=1$, $n=1$) if we throw away the decoupled hypermultiplet $(\Tr\,\Phi,\Tr\,\tilde\Phi)$, the supersymmetry actually enhances to $(4,4)$ and we expect to have a $(4,4)$ sigma model with $\tilde{X}=\mathbb{C}^2/\mathbb{Z}_2$ target space. It follows that $E$ is isomorphic to the tangent bundle $T\tilde{X}$. The resulting $(4,4)$ SCFT has central charges $\tilde{c}_L=\tilde{c}_R=6$.

\subsection{$\CN=(4, 4)$ theories}

Most of the statements about $(0,4)$ theories made in previous sections also hold for their $(4,4)$ counterparts. The main difference is that now the theory also has a Coulomb branch (and in the case of $SU(2)$ gauge group there is no FI parameter to switch it off) that receives quantum corrections.

Let us replace all $(0,4)$ hypermultiplets by $(4,4)$ multiplets and $(0,4)$ vector multiplets by $(4,4)$ vector multiplets in quiver notations (\ref{fig:SU2-quiver-notations}). Then $(4,4)$ analogs of (\ref{index-T2}) and (\ref{index-V2}) read
\begin{equation}
 \CI^{(4,4)}_{T_2}(x,y,z;v;q)\equiv \frac{\theta(q^{1/2}u\,x^\pm y^\pm z^\pm)}{\theta(v\,x^\pm y^\pm z^\pm)}\,,
 \label{index-T2-44}
\end{equation}
\begin{equation}
 \CI^{(4,4)}_{V,SU(2)}(\xi;v,q)\equiv \frac{\theta(q/v^2)}{{\theta(q^{1/2}uv^{-1})\theta(q^{1/2}u^{-1}v^{-1})}}\frac{\theta(q\,\xi^{\pm2}/v^{2})\theta(\xi^{\pm2})(q;q)^2}{\theta(q^{1/2}uv^{-1}\,\xi^{\pm2})\theta(q^{1/2}u^{-1}v^{-1}\,\xi^{\pm2})}
 \label{index-V2-44}
\end{equation}
where $u$ is the fugacity for the additional $SU(2)$ R-symmetry of $\CN=(4,4)$ UV superalgebra. In particular, the index of the $(4,4)$ theory corresponding to the quiver in Fig. \ref{fig:SU2-schannel},
\begin{equation}
 \CI^{(4,4)}_{\schannel}(x,y,z,w;v;q)=\half\int\limits_\text{JK} \frac{d\xi}{2\pi i\xi}\, \CI^{(4,4)}_{T_2}(x,y,\xi;v;q)\,\CI^{(4,4)}_{V,SU(2)}(\xi;v,q)\,\CI^{(4,4)}_{T_2}(1/\xi,z,w;v,q),
 \label{I44SO8int}
\end{equation}
also satisfies the crossing symmetry property
\begin{equation}
  {\CI}^{(4,4)}_{\schannel}(x,y,z,w;v;q)-{\CI}^{(4,4)}_{\schannel}(x,z,y,w;v;q)=0
  \label{index-44-crossing}
\end{equation}
which means that similarly to the $(0,4)$ case one can use (\ref{index-T2-44}) and (\ref{index-V2-44}) to define a 2d TQFT.

%%%%%%%%%%%%%%%%%%%%%%%%%%%%%%%%%%%%%%%%%%%%%%%%%%%%%%%%%%%%%%
\section{$SU(N)$ theories}
In this section we study $\CN=(0, 4)$ quiver theories with $SU(N)$ gauge group. In section \ref{subsec:SUN}, we consider a $SU(N)$ version of the SQCD with $\CN=(0, 4)$ and $\CN=(4, 4)$ supersymmetry. We find a crossing-symmetry of the elliptic genus for this case as well. In section \ref{subsec:TN}, we argue for the existence of 2d analog of the $T_N$ theory.

\subsection{$SU(N)$ with $2N$ flavors and its crossing symmetry} \label{subsec:SUN}

Let us consider the $\CN=(0, 4)$ $SU(N_c)$ gauge theory with $N_f$ fundamental hypermultiplets. The following table lists the $(0,2)$ superfields of the theory and their charges w.r.t. various symmetry groups:
\be
\begin{tabular}{c|ccc|cc}
 & $SU(N_c)$ & $SU(N_f)$ & $U(1)_B$ & $U(1)_R^- \times U(1)_R^+$ & $U(1)_v$ \\
 \hline
 $\Theta$ & adj & $1$ & $0$ & $(-1, 1)$ & $-2$\\
 $\Phi$ & $N_c$ & $N_f$ & $1$ & $(1, 0)$ & $1$ \\
 $\tilde{\Phi}$ & $\bar{N}_c$ & $\bar{N}_f$& $-1$ & $(1, 0)$ & $1$ \\
\end{tabular}
\label{SUN-table}
\ee
where $U(1)_R^- \times U(1)_R^+ \subset SU(2)_R^- \times SU(2)_R^+$, $U(1)_v = U(1)_R^- - U(1)_R^+$, and $U(1)_B$ is the barionic $U(1)$ symmetry. The theory has the following superpotential
\be
 W = \tilde{\Phi} \Theta \Phi \ ,
\ee
necessary to ensure $\CN=(0, 4)$ supersymmetry.

The gauge anomaly coefficient is given by (see appendix \ref{appendixanomalies}):
\be
 k_{SU(N_c)} = \half  2  N_f - N_c - N_c
  = N_f - 2N_c \ ,
\ee
which implies that we should take $N_f = 2N_c \equiv 2N$.
The anomaly coefficients for the flavor $SU(N_f)$ symmetry and $U(1)_B$ are
\be
 k_{SU(N_f)} = N \ , \qquad k_{U(1)_B} =  4 N^2 \ .
\ee
Also, the theory has non-vanishing 't Hooft anomalies involving $U(1)_v$:
\be
k_{U(1)_v} = 4 \ ,  \qquad k_{U(1)_R^+ \,\cdot\, U(1)_v} = -2 \ .
\ee

Similarly to the case with $SU(2)$ gauge group considered in the previous section, the theory has only Higgs branch and we expect $SU(2)_R^+$ to be the R-symmetry of the SCFT at the IR fixed point. By counting its anomaly coefficient in the UV theory we obtain
\be
 c_R = 6(N^2+1)\ , \qquad  c_L = 4(N^2+1) \ .
\ee
Again, $c_R/6$ agrees with the quaternionic dimension of the Higgs branch as expected.

\begin{figure}[ht]
\centering
\includegraphics[scale=1]{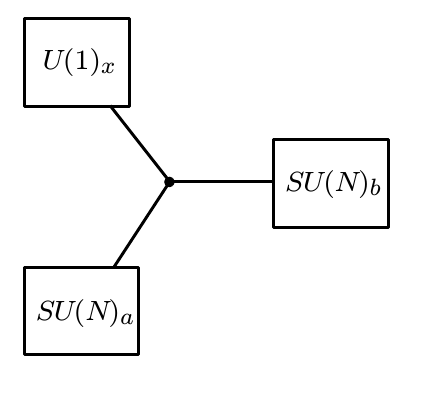}
\caption{The quiver diagram for theory $U_N^{(0, 4)}$ of free hypermultiplet in the bifundamental representation of $SU(N)_a \times SU(N)_b$ and barionic symmetry $U(1)_x$. }
\label{fig:UN}
\end{figure}

As in section \ref{section:SU2basic} we find that the index of the theory has a similar crossing-symmetry property. Consider a trinion $U_N^{(0, 4)}$ describing a hypermultiplet in the bifundamental representation of $SU(N) \times SU(N)$ (see Fig. \ref{fig:UN}). It also has a baryonic symmetry $U(1)$. The index is given by
\be
 \CI^{(0, 4)}_{U_N} (\ba, \bb, x; v; q) = \prod_{i, j=1}^N \frac{1}{\theta(v (x a_i b_j)^\pm)} \ ,
\ee
where $\ba, \bb, x$ denote fugacities for $SU(N)_a \times SU(N)_b \times U(1)_x$ respectively. Now, let us glue a pair of $U_N^{(0,4)}$ (by coupling them both to a $(0,4)$ $SU(N)$ vector multiplet) to form $SU(N)$ SQCD with $2N$ flavors. The index of the resulting theory reads
\be
 \CI^{(0, 4)}_{\schannel} (\ba, \bb, x, y) =\frac{1}{N!} \int_{\textrm{JK}} \left( \prod_{i=1}^{N-1} \frac{d \xi_i}{2\pi i \xi_i} \right) \CI^{(0, 4)}_{U_N} (\ba, \bxi, x) \CI^{(0, 4)}_{V, SU(N)}(\bxi) \CI^{(0, 4)}_{U_N} (\bxi^{-1}, \bb, y) \ ,
\ee
where we dropped $v, q$ dependence in the expression for brevity. The vector multiplet index is given by
\be
 \CI_{V, SU(N)}^{(0, 4)} (\bxi; v; q) = \theta\left(\frac{q}{v^2}\right) \prod_{i \neq j} \theta \left(\frac{q}{v^2}  \frac{\xi_i}{ \xi_j} \right) \theta \left(\frac{\xi_i }{ \xi_j}\right) \ .
\ee
Here we have used the flavor fugacities with $SU(N)_a \times SU(N)_b \times U(1)_x \times U(1)_y \subset SU(2N) \times U(1)$ manifest.

We find that the index is invariant under the exchange of $\mathbf{a} \leftrightarrow \mathbf{b}$ or equivalently $x \leftrightarrow y$:
\be
 \CI^{(0, 4)}_{\schannel} (\ba, \bb, x, y) = \CI^{(0, 4)}_{\schannel} (\bb, \ba, x, y) = \CI^{(0, 4)}_{\schannel} (\ba, \bb, y, x) \ .
\ee
\begin{figure}[ht]
\centering
\includegraphics[scale=1]{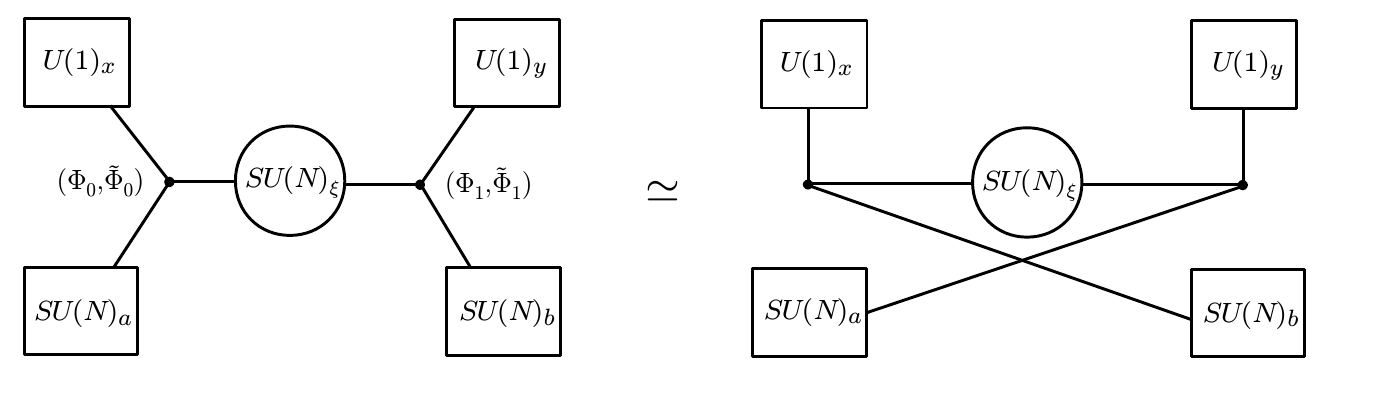}
\caption{The quiver on the left represents $(0,4)$ $SU(N)$ SQCD with $2N$ flavors as a gluing of two copies of $U_N^{(0,4)}$. The equivalence to the diagram on the right represents crossing-symmetry of the index.}
\label{fig:SUN-crossing}
\end{figure}
On the level of quiver diagrams this can be understood as a crossing symmetry between $s$-channel and $u$-channel (see Fig. \ref{fig:SUN-crossing}). This duality or crossing-symmetry implies that the spectrum of the operators in the CFT should obey such property. It is not automatic from the global symmetry of the theory.

The crossing-symmetry can be understood as a duality. Even though the matter content on both side of the dual theories are the same, the operator contents on one side are mapped to another operators on the other side. For example, we have gauge-invariant operators of the form as in the following table (here we decomposed $(\Phi, \tilde{\Phi})$ from (\ref{SUN-table}) into $(\Phi_{0, 1}, \tilde{\Phi}_{0, 1})$ of two copies of $U_N^{(0,4)}$ as shown in Fig. \ref{fig:SUN-crossing}):
\be
\begin{tabular}{c|cccc}
 operators & $U(1)_x$ & $U(1)_y$ & $SU(N)_A$ & $SU(N)_B$ \\
 \hline
$\e (\Phi_0)^k (\tilde{\Phi}_1)^{N-k}$ & $k$ & $-N+k$ & $\Lambda^k $ & $\Lambda^{k}$\\
 $\e (\tilde{\Phi}_0)^k (\Phi_1)^{N-k}$ & $-k$ & $N-k$ & $\Lambda^{N-k} $ & $\Lambda^{N-k}$ \\
 $ \Phi_0 \tilde{\Phi}_0$ & $0$ & $0$ & $N \otimes \bar{N}$ & $1$  \\
 $ \Phi_1 \tilde{\Phi}_1$ & $0$ & $0$ & $1$ & $N \otimes \bar{N}$ \\
 $ \Phi_0 \Phi_1$ & $1$ & $1$ & $N$ & $\bar{N}$  \\
  $ \tilde{\Phi}_0 \tilde{\Phi}_0$ & $-1$ & $-1$ & $\bar{N}$ & $N$
\end{tabular}
\ee
where $\Lambda^k$ is $k-$th antisymmetric representation and $\e$ is completely antisymmetric tensor to contract the gauge indices. The first two lines are baryonic operators where as the latter four are mesonic operators. Under the exchange of $U(1)_x$ and $U(1)_y$, the mesonic operators remain unchanged, but the baryonic operators are mapped via
\be \label{eq:SUNdualMap}
 (\Phi_0)^k (\tilde{\Phi}_1)^{N-k} \to (\Phi_1)^k (\tilde{\Phi}_0)^{N-k} \ , \quad \textrm{and} \qquad (\tilde{\Phi}_0)^k (\Phi_1)^{N-k} \to \tilde{\Phi}_1^k (\Phi_0)^{N-k} \ .
\ee

Let us now consider the $\CN=(4, 4)$ version of the theory. The matter contents are essentially the same except that we replaced $(0, 4)$ multiplets to $(4, 4)$ multiplets. We can write it more explicitly in terms of $\CN=(0, 2)$ superfields as in the following table:
\be
\begin{tabular}{c|ccc|c}
 & $SU(N_c)$ & $SU(N_f)$ & $U(1)_B$ & $U(1)_R^- \times U(1)_R^+ \times U(1)_I$ \\
 \hline
 $\Theta$ & adj & $1$ & $0$ & $(-1, 1, 0)$ \\
 $\S$ & adj & 1 & 0 & $(0, 1, 1)$ \\
 $\tilde\S$ & adj & 1 & 0 & $(0, 1, -1)$ \\
 $\Phi$ & $N_c$ & $N_f$ & $1$ & $(1, 0, 0)$ \\
 $\tilde{\Phi}$ & $\bar{N}_c$ & $\bar{N}_f$& $-1$ & $(1, 0, 0)$ \\
 $\G$ & $N_c$ & $N_f$ & $1$ & $(0, 0, 1)$ \\
 $\tilde{\G}$ & $\bar{N}_c$ & $\bar{N}_f$ & $-1$ & $(0, 0, 1)$
\end{tabular}
\ee
where $SU(2)_R^- \times SU(2)_R^+ \times SU(2)_I$ is $\mathcal{N}=(4,4)$ R-symmetry which an extra $SU(2)_I$ factor compared to the $\mathcal{N}=(0,4)$ case. As discussed in appendix \ref{appendix04}, this R-symmetry can be understood from the dimensional reduction of 6d $\CN=(1, 0)$ multiplets.
The theory have the following $J$-type superpotential and $E$-terms:
\be
 W = \tilde{\Phi}\Theta \Phi + \tilde{\Gamma} \tilde{\S} \Phi + \tilde{\Phi} \tilde{\S} \G \ ,
\ee
\be
 E_\Theta = [\S, \tilde{\S}] \ , \quad E_\G = \S \Phi \ ,  \quad E_{\tilde{\G}} = -\tilde{\Phi} \S \ .
\ee
The $\CN=(4, 4)$ gauge theory is expected to flow to two distinct CFTs on the Higgs branch and on the Coulomb branch \cite{Witten:1997yu, Aharony:1999dw}.

We can also compute the index for this theory. The index for the trinion theory $U_N^{(4, 4)}$ consists of the free bifundamental $(4, 4)$ hypermultiplets can be written as
\be
 \CI^{(4, 4)}_{U_N} (\ba, \bb, x; u, v; q) = \prod_{i, j=1}^N \frac{\theta(q^{1/2}u (x a_i b_j)^\pm)}{\theta(v (x a_i b_j)^\pm)} \ ,
\ee
where $u$ is the fugacity for the $U(1)_I \subset SU(2)_I$ symmetry. The vector multiplet index reads
\be
 \CI_{V, SU(N)}^{(4, 4)} (\bxi; u, v; q) = \left( \frac{\theta(q v^{-2})}{\theta(q^\half u^\pm v^{-1}) } \right)^{N-1} \prod_{i \neq j} \frac{\theta (q v^{-2} {\xi_i}/{ \xi_j} ) \theta ({\xi_i }/{ \xi_j})}{\theta (q^\half u^\pm v^{-1} {\xi_i}/{ \xi_j} )} \ .
\ee
Now we can write the index for the SQCD as
\be
 \CI^{(4, 4)}_{\schannel} (\ba, \bb, x, y) = \frac{1}{N!} \int_{\textrm{JK}} \left( \prod_{i=1}^{N-1} \frac{d \xi_i}{2\pi i \xi_i} \right) \CI^{(4, 4)}_{U_N} (\ba, \bxi, x) \CI^{(4, 4)}_{V, SU(N)}(\bxi) \CI^{(4, 4)}_{U_N} (\bxi^{-1}, \bb, y) \ ,
\ee
where we suppressed the dependence on $u, v$ and $q$. It also satisfies the crossing symmetry
\be
 \CI^{(4, 4)}_{\schannel} (\ba, \bb, x, y) = \CI^{(4, 4)}_{\schannel} (\bb, \ba, x, y) = \CI^{(4, 4)}_{\schannel} (\ba, \bb, y, x) \ ,
\ee
which implies constraints on the operator spectrum and IR duality as in the $\CN=(0, 4)$ case.

%%%%%%%%%%%%%%%%%%%%%%%%%%%%%%%%%%%%%%%%%%%%%%
\subsection{Dualities of quiver theories and $T_N^{(0, 4)}$ theory} \label{subsec:TN}
In this section, we discuss quiver gauge theories and dualities.

\subsubsection{Quiver gauge theories}

\paragraph{Linear quiver}
Let us consider linear quiver theories composed of connecting $m$ copies of $U_N$ blocks. This will yield $SU(N)^{m-1}$ gauge theory with bifundamentals in $SU(N)_i \times SU(N)_{i+1}$ where we identify $SU(N)_0$ and $SU(N)_{m}$ as the global symmetry groups, see Fig. \ref{fig:SUN-linear-quiver}.

\begin{figure}[ht]
\centering
\includegraphics[scale=1]{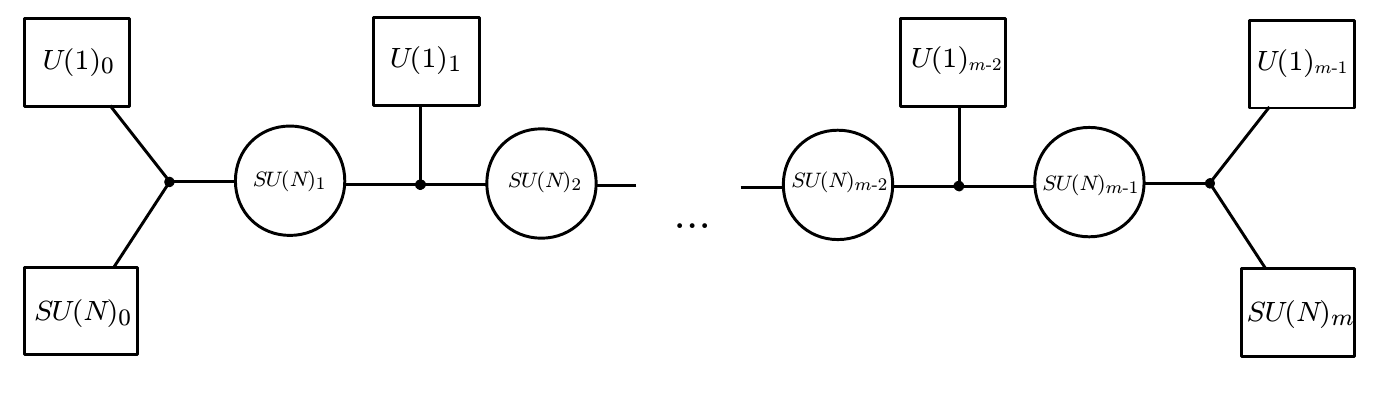}
\caption{A linear quiver realizing a theory with $SU(N)^{m-1}$ gauge group and $SU(N)^2\times U(1)^{m}$ flavor group.}
\label{fig:SUN-linear-quiver}
\end{figure}

The quiver gauge theory flows to CFT on the Higgs branch. The central charges can be computed easily to be
\be
 c_R = 6 \left(N^2 + m - 1 \right)  \ , \quad c_L = 4 (N^2 + m - 1) \ .
\ee
The (quaternionic) dimension of the Higgs branch is given by $c_R/6$.

As we have discussed in section \ref{subsec:SUN}, the index of the quiver theory also enjoys crossing-symmetry. It can be also applied to the linear quiver theory, which has the global symmetry $SU(N)_A \times SU(N)_B \times \left( \prod_{i=1}^m U(1)_i \right)$. The crossing-symmetry now extends to the permutation of all the $U(1)_i$ symmetries. Therefore we have a duality map analogous to \eqref{eq:SUNdualMap}, by applying the duality repeatedly. The single-trace gauge invariant operators contains the bayonic operators $\det \Phi_i$ and $\det \tilde{\Phi}_i$ with $i=0, \cdots, m$ and mesonic operators $\Phi_0 \tilde{\Phi}_0$ and $\Phi_m \tilde{\Phi}_m$. Under the permutation, $U(1)_i \leftrightarrow U(1)_j$, we exchange $\det \Phi_i \leftrightarrow \det \Phi_j$.

\paragraph{Circular quiver}

\begin{figure}[ht]
\centering
\includegraphics[scale=1]{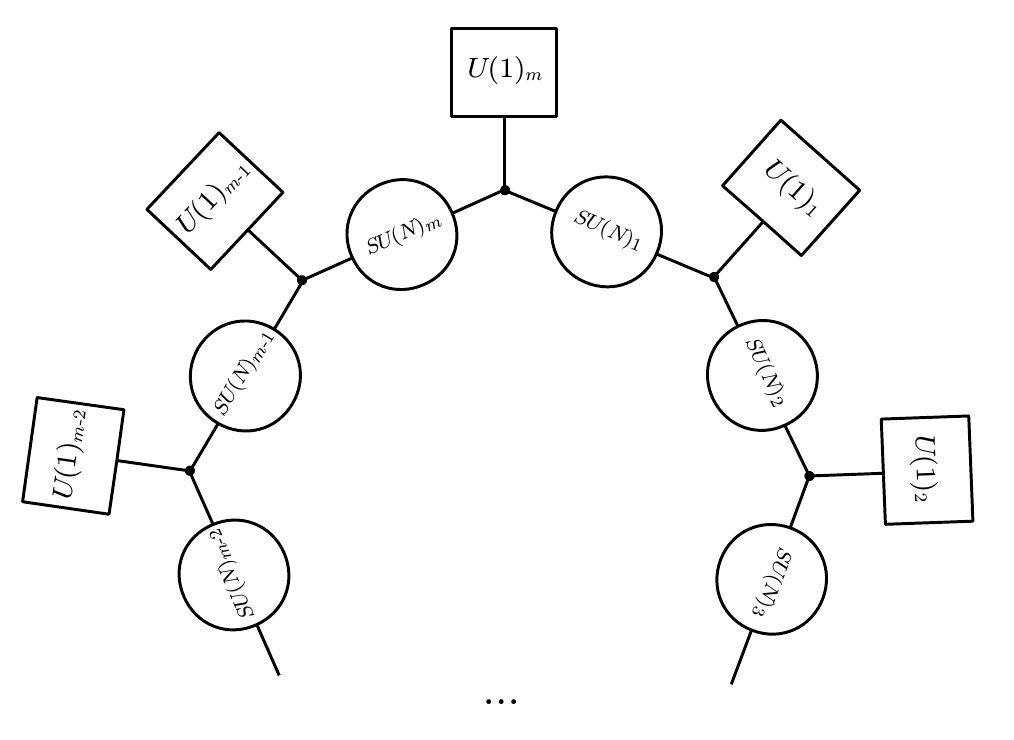}
\caption{A circular quiver realizing a theory with $SU(N)^{m}$ gauge group and $U(1)^{m}$ flavor group.}
\label{fig:SUN-circular-quiver}
\end{figure}

We can also consider a circular quiver theory by gauging the diagonal subgroup of $SU(N)_0 \times SU(N)_m$ of the linear quiver. As in the case of $SU(2)$ theories, we get a CFT on the Kibble branch with dimension $m+1$, see Fig. \ref{fig:SUN-linear-quiver}. The central charge of this theory is given by
\be
 c_R = 6(n_h - n_v + 1) = 6(m+1) \ , \qquad c_L = 4(m+1) + 2 \ .
\ee
Note that the central charges do not depend on the choice of the gauge group, even though the elliptic genus does depend on the gauge group.

\subsubsection{Analog of Argyres-Seiberg duality and $T^{(0, 4)}_3$ theory}
\label{section:T3}

Let us consider the $SU(3)$ case.
\begin{figure}[ht]
\centering
\includegraphics[scale=0.85]{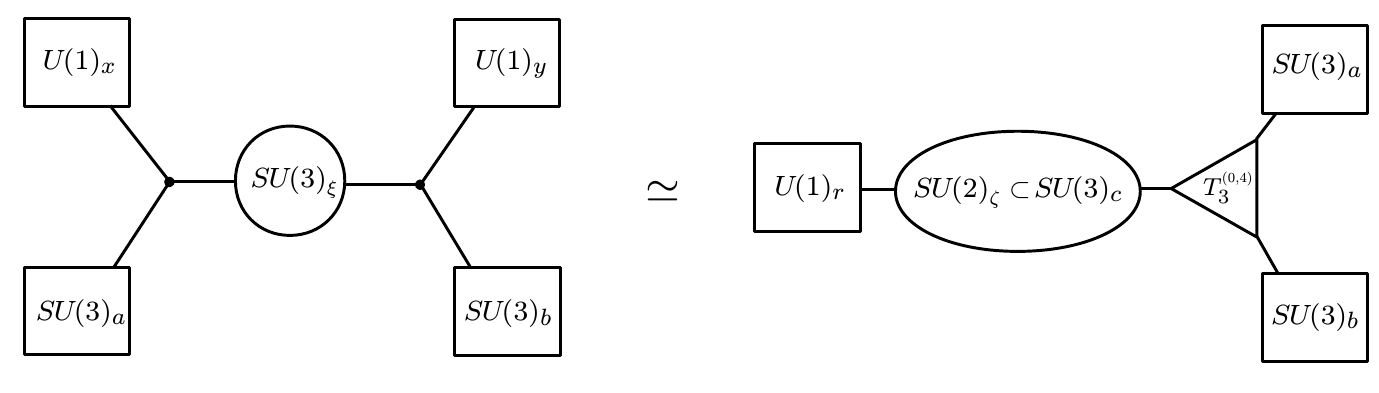}
\caption{Two-dimensional $\mathcal{N}=(0,4)$ analog of Argyres-Seiberg duality. The subscripts of flavor and gauge groups denote corresponding fugacities in the index.}
\label{fig:T3-duality}
\end{figure}
Similarly to the $\CN=2$ 4d case \cite{Argyres:2007cn} we conjecture that $SU(3)$ gauge theory with $6$ flavors is dual to the theory constructed from $T_3^{(0,4)}$, two hypermultiplets and $(0,4)$ $SU(2)$ vector multiplet gauging the diagonal of $SU(2)\subset SU(3)$ subgroup of the flavor symmetry $T_3^{(0,4)}$ and $SU(2)$ flavor symmetry acting on two hypermultiplets (see Fig. \ref{fig:T3-duality}). On the level of indices the duality reads

\begin{equation}
 I^{(0,4)}_\schannel (\ba,\bb;x,y)=\frac{1}{2}\int\limits_\text{JK}\frac{d\zeta}{2\pi i\zeta}\,
 \frac{I^{(0,4)}_{V,SU(2)}(\zeta)}{\theta(vs^{\pm 1}\zeta^{\pm 1})}\,
 I_{T_3}^{(0,4)}(\ba,\bb,\bc)\,,
 \label{T3-duality-index}
\end{equation}
\[
(c_1,c_2,c_3)\equiv(r\zeta,r/\zeta,1/r^2),\qquad x\equiv s^{1/3}/r,\qquad y\equiv s^{-1/3}/r
\]
Assuming that as in $SU(2)$ case $T^{(0,4)}_3$ describes a certain Higgs branch CFT its central charges can be easily determined from the relation depicted in Fig. \ref{fig:T3-duality}:
\begin{equation}
 c_R=6\cdot 11\,\qquad c_L=4\cdot 11,
\end{equation}
where 11 is the quaternionic dimension of the Higgs branch.

Similarly to $\mathcal{N}=2$ 4d case \cite{Gadde:2010te} one can go further and solve the integral equation (\ref{T3-duality-index}) for $I^{(0,4)}_{T_3}$. To do so let us use expression (\ref{index-V2}) for $I^{(0,4)}_{V,SU(2)}(\zeta)$ and apply the inversion formula (\ref{inv-formula}):
\begin{equation}
 I_{T_3}^{(0,4)}(\ba,\bb,\bc)=\frac{(q;q)^2}{2\,\theta(v^2\zeta^{\pm 2})}
 \int\limits_\text{JK}\frac{d s}{2\pi i\,s}\,
  \frac{\theta(s^{\pm 2})\theta(v^{-2})}{\theta(v^{-1}s^{\pm 1}\zeta^{\pm 1})}\,
 I^{(0,4)}_\schannel (\ba,\bb;x,y)\,.
 \label{T3-index-formula}
\end{equation}
Since at each step the one can calculate contour integrals explicitly by residues, this provides us with explicit (although quite long) expression for the index of $T_3^{(0,4)}$ theory. The result is symmetric under permutation of $SU(3)$ fugacities $\ba,\bb,\bc$ which is a non-trivial check supporting the conjecture about the existence of such theory $T_3^{(0,4)}$ and the fact that its flavor symmetry is enhanced to $E_6\supset SU(3)^3$. The expansion of the index w.r.t $q$ and $v$ in terms of characters of $E_6$ representations reads:
\begin{equation}
\begin{split}
 I_{T_3}^{(0,4)}=&\left(\mathbf{1}+\mathbf{78}\,v^2+\mathbf{2430}\,v^4+\ldots\right)\\
 &+\left(\mathbf{(1+78)}+(\mathbf{1} + 2 \cdot\mathbf{78 + 2430 + 2925})v^2+\ldots\right)\,q
 +\ldots
 \label{I04E6series}
\end{split}
\end{equation}
Let us note that $q^0$ order coincides with the Hilbert series of the Higgs branch moduli space, conjectured to be the same as the moduli space of one $E_6$ instanton \cite{Benvenuti:2010pq,Keller:2011ek,Keller:2012da}. The leading terms also agree with the $S^2 \times T^2$ partition function computed in \cite{Gadde:2015xta}.

The $T_3^{(0, 4)}$ is a 2d version of the celebrated $E_6$ SCFT of Minahan-Nemeschansky \cite{Minahan:1996fg}. One important difference here is that our theory does not have any Coulomb branch. We can also come up with a ``Lagrangian" for the ``non-Lagrangian" $E_6$ SCFT as done in \cite{Gadde:2015xta}. The $\CN=(0,2)$ field content can be straightforwardly read off the integral representation of the index of $T_3^{(0,4)}$. Namely, (\ref{T3-index-formula}) represents combining the theory associated to the quiver in the left part of Fig. \ref{fig:T3-duality} together with two chiral multiplets in representations
\begin{equation}
 \mathbf{(2,2)_{-1}\oplus (1,3)_{2}}
\end{equation}
of $SU(2)_s\times SU(2)_\zeta\times U(1)_v$, two Fermi multiplets in
\begin{equation}
  \mathbf{(1,1)_{-2}\oplus (1,1)_{2}}\,,
\end{equation}
and then gauging $SU(2)_s$ with $\mathcal{N}=(0,2)$ Vector multiplet. The choice of superpotential should be consistent with global symmetry charges appearing in the index. The result is in agreement with twisted compactification of $\CN=1$ 4d theory proposed in \cite{Gadde:2015xta} on $S^2$.

As we have discussed in section \ref{subsec:SU2TQFT}, crossing-symmetry implies the TQFT structure of the elliptic genus. But unlike the case of $SU(2)$ theories, we have two distinct type of punctures: $SU(3)$ (maximal) puncture and $U(1)$ (minimal) puncture. We have already shown in section \ref{subsec:SUN} that the index remains unchanged upon exchanging two $U(1)$ punctures or two $SU(N)$ punctures in the second frame of figure \ref{fig:T3-duality}. With the expansion \ref{I04E6series} we can further show that crossing-symmetry exists in the theory with four maximal punctures up to certain order of $q$ and $v$. Therefore the TQFT structure holds for the $SU(3)$ theories as well.

\subsubsection{$T^{(0, 4)}_N$ theory and duality}

So far we have discussed 2d $\CN=(0, 4)$ gauge theories without referring to its higher-dimensional origin. Let us point out that theories we studied so far can be realized from M5-branes on a product Riemann surfaces.
Consider 4d $\CN=2$ class $\CS$ theory of type $A_{N-1}$ with the UV curve given by $\CC$ with genus $g$ and $n$ punctures.
Now, let us compactify this 4d theory on $\mathbb{CP}^1$ with a partial topological twist. Since we have two independent R-symmetries $SU(2)_R \times U(1)_r$, we have to choose one. Twisting with respect to $SU(2)_R$ and $U(1)_r$ gives us $\CN=(2, 2)$ or $\CN=(0, 4)$ supersymmetry in 2d respectively. We are interested in the $\CN=(0, 4)$ twisting. In this case, for each free vector multiplets in 4d, we get one $(0, 4)$ vector, and for each free hypermultiplets in 4d, we get one $\CN=(0, 4)$ hypermultiplet. See appendix \ref{sec:twist} for the detail.

Upon taking small volume limit of $\mathbb{CP}^1$, we also take the 4d gauge coupling to be small to get a 2d gauge theory, since $1/g_{2d}^2 = \textrm{vol}(\mathbb{CP}^1)/g_{4d}^2$. There can be also S-dual descriptions for the 4d theory, which we also dimensionally reduce to another 2d gauge theory. Note that for this case, we need to take the dual gauge couplings to zero while shrinking the volume of the sphere. In principle, dimensional reduction of these two different limits do not necessarily give the same CFT in 2d. When taking the 2d limit, we have to decouple 4d building blocks in a different way for each S-dual frames. From there we are turning on gauge couplings to RG flow to 2d CFT, which we call as $T_{\frak{su}(N)}[\mathbb{CP}^1 \times \CC_{g, n}]$. Nevertheless, we find evidences that different 2d `gauge theories' (which can also involve `non-Lagrangian' $T_N^{(0, 4)}$ block) obtained from dual descriptions flow to the same 2d $\CN=(0, 4)$ SCFT.\footnote{See discussions on 3d to 2d \cite{Aganagic:2001uw} and 4d to 3d reduction \cite{Aharony:2013dha}.} Note that since the gauge couplings undergo RG flows, the dependence on the complex structure of $\CC_{g, n}$ disappears in the IR. Crossing-symmetry (or TQFT structure) of elliptic-genus is a check of this conjecture.

As a corollary, the effective number of vector and hypermultiplets remain the same in the 2d $\CN=(0, 4)$ theory as the 4d $\CN=2$ theory. Given this assumption, we can compute the central charges of the 2d theory $T_{\frak{su}(N)}[\mathbb{CP}^1 \times \CC_{g, n}]$. The number of effective vector and hypermultiplets can be decomposed in terms of a contribution from the background Riemann surface, and local contributions from the punctures \cite{Chacaltana:2012zy}. For the $SU(N)$ theory, we get
\be
 n_h (\CC_g) = \frac{4}{3} (g-1)  N (N^2 - 1) \ , \quad
 n_v (\CC_g) = \frac{1}{3} (g-1) (N-1)(4N^2 + 4N + 3) \ ,
\ee
for a genus $g$ curve, and
\be
 n_h (Y_{\textrm{max}}) = \frac{2}{3} N(N^2-1) \ , \qquad
 n_v (Y_{\textrm{max}}) = \frac{1}{6} N(N-1)(4N+1)\  ,
\ee
for the maximal puncture and
\be
 n_h (Y_{\textrm{min}}) = N^2 \ , \qquad
 n_v (Y_{\textrm{min}}) = (N+1)(N-1) \  ,
\ee
for the minimal puncture. We define $n^{(g, n)}_h = n_h(\CC_g) + \sum_{i=1}^n n_h (Y_i)$ and $n^{(g, n)}_v = n_v(\CC_g) + \sum_{i=1}^n n_v (Y_i)$.

As we have discussed, for $g=0$, we have the Higgs branch, and for $g \ge 1$, we have the Kibble branch. We get
\be
 c_R = 6(n_h^{(g=0, n)}-n_v^{(g=0, n)}) \ , \quad
 c_L = 4(n_h^{(g=0, n)} - n_v^{(g=0, n)})\ ,
\ee
for $g=0$ and
\be
 c_R = 6(n_h^{(g, n)}-n_v^{(g, n)}+g) \ , \quad
 c_L = 4(n_h^{(g, n)} - n_v^{(g, n)}+g) + 2g \ ,
\ee
for $g\ge 1$.  One can check that this result indeed agrees with central charge expressions we computed in previous sections from the 2d gauge theory description for the case with $g=0$ with 2 maximal and $n-2$ minimal punctures and $g=1$ with $n$ minimal punctures.

The $T_N^{(0, 4)}$ theory corresponds to a sphere with 3 maximal punctures with $SU(N)_a \times SU(N)_b \times SU(N)_c$ global (non-R) symmetry. We get the central charges to be
\be
 c_R = 3(N-1)(3N+2) \ , \qquad c_L = 2(N-1)(3N+2) \ ,
\ee
agrees with $N=2, 3$ results in section \ref{section:SU2basic} and \ref{section:T3}.

\begin{figure}[ht]
\centering
\includegraphics[scale=0.85]{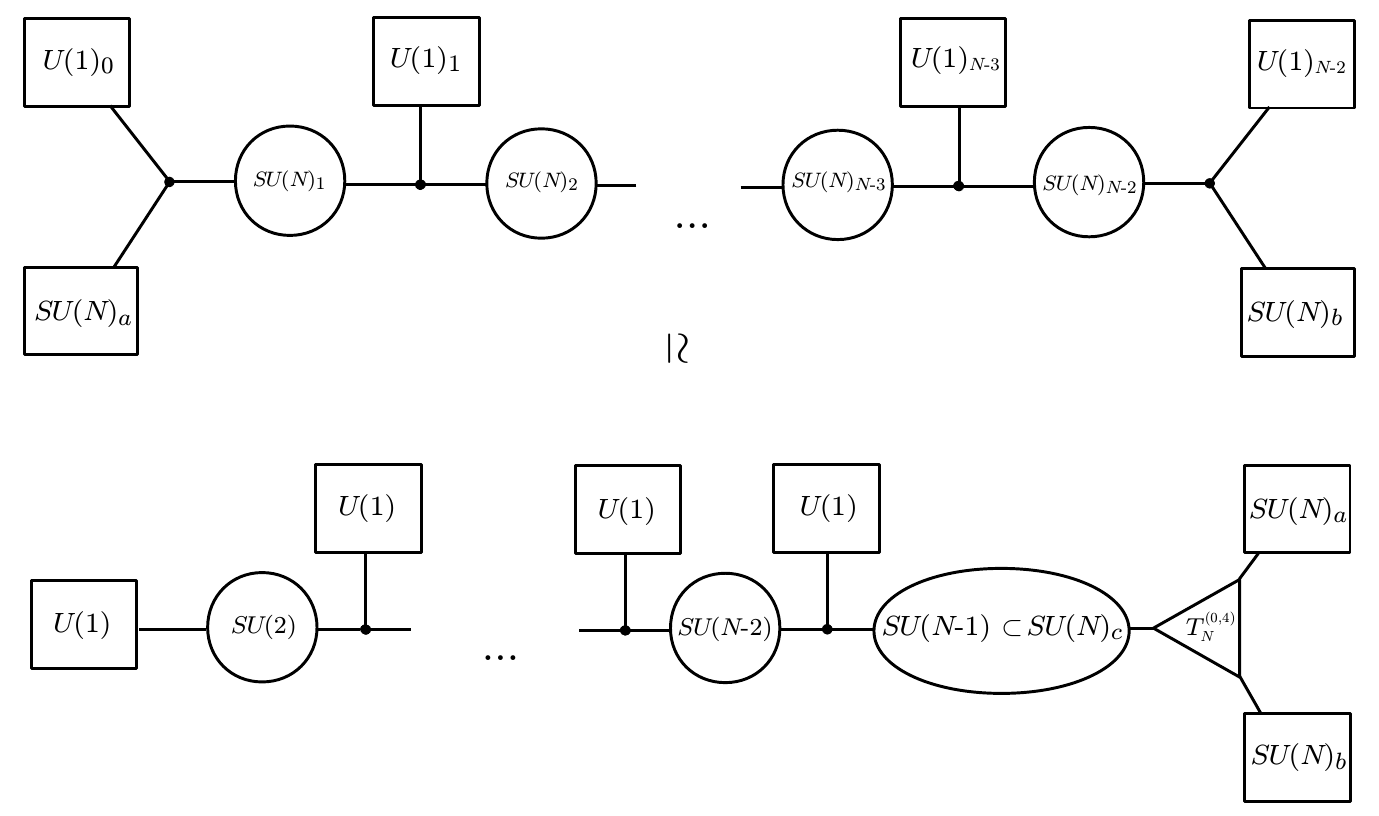}
\caption{The duality between $T_N^{(0,4)}$ coupled to a quiver tail (bottom) and a linear quiver with $SU(N)^{N-2}$ gauge group (top).}
\label{fig:TN-duality}
\end{figure}

We can also compute the central charges from the dual Lagrangian description. When $T_N$ theory is coupled to a quiver tail, of the form $SU(N)_c \supset SU(N-1) \times SU(N-2) \times \cdots \times SU(2)$ with bifundamentals and fundamentals attached as in the quiver diagram in the bottom of Fig. \ref{fig:TN-duality}. This theory is dual to a linear quiver with gauge group $SU(N)^{N-2}$, and fundamental attached to the end as in the top of Fig. \ref{fig:TN-duality}. The $SU(N)$ flavor symmetry anomaly coefficient can be computed in the dual frame:
\be
 k_{SU(N)_x} = \tr \g^3 SU(N)_x^2 = N \qquad (\textrm{where}~x=a, b, c) \ .
\ee

%%%%%%%%%%%%%%%%%%%%%%%%%%%%%%%%%%%%%%%%%%%%%%%%
\section{Other dualities}

\subsection{$\CN=(0,2)$ and $\CN=(2,2)$ analog of the crossing symmetry}

In this section we will show that there are $\CN=(0,2)$ and $\CN=(2,2)$ analogies of the crossing symmetry property of the spectrum considered in the previous section. In what follows we will study the cases  $\CN=(0,2)$ and $\CN=(2,2)$ in parallel. Let us define $U_N^\CN$ as $N^2$ chiral multiplets in $\mathbf{(N,N)_{+1}}$ representation of $SU(N)_a\times SU(N)_b\times U(1)_x$ flavor symmetry. The corresponding index contribution reads
\begin{equation}
 \CI_{U_N}^{(0,2)}(\ba,\bb,x;q)=\prod_{i,j=1}^N\frac{1}{\theta(xa_ib_j)}
\end{equation}
or
\begin{equation}
 \CI_{U_N}^{(2,2)}(\ba,\bb,x;q)=\prod_{i,j=1}^N\frac{\theta(txa_ib_j)}{\theta(xa_ib_j)}
\end{equation}
where $\ba=\{a_i\}_{i=1}^N$, $\bb=\{b_i\}_{i=1}^N$ are $SU(N)_{a,b}$ fugacities satisfying
\begin{equation}
 \prod_{i}a_i=1,\qquad \prod_{i}b_i=1,
\end{equation}
and $x$ is $U(1)_x$ fugacity. In the $\CN=(2,2)$ case we have an extra left-moving $U(1)$ R-symmetry fugacity $t$. Now let us consider $\CN=(0,2)$ or $\CN=(2,2)$ $SU(N)$ SQCD with $N$ fundamental and $N$ anti-fundamental flavors, which can be obtained by coupling two copies of $U_N^\CN$ to $SU(N)$ vector multiplet. In the $\CN=(0,2)$ case, similarly to the $(0,4)$ case, gauge anomaly contributions from chiral and vector multiplets cancel each other. The theory has the following index:
\begin{equation}
  \CI^\CN_{\schannel} (\ba, \bb, x, y) = \frac{1}{N!}\int\limits_{\textrm{JK}}\prod_{i=1}^{N-1} \frac{d\xi_i}{2\pi i\,\xi_i}\,
  \CI^\CN_{U_N} (\ba, \bxi, x) \,
  \CI^\CN_{V, SU(N)}(\bxi)\,
  \CI^\CN_{U_N} (\bxi^{-1}, \bb, y) \ ,
  \label{index-N2-int}
\end{equation}
where
\begin{equation}
 \CI^{(0,2)}_{V,SU(N)}(\bxi)= (q;q)^{N-1} \prod_{i\neq j}\theta(\xi_i/\xi_j)\,,
 \label{index-02-vector}
\end{equation}
\begin{equation}
 \CI^{(2,2)}_{V,SU(N)}(\bxi)= (q;q)^{N-1}\,\frac{\prod_{i\neq j}\theta(\xi_i/\xi_j)}{\prod_{i,j}\theta(t\,\xi_i/\xi_j)}\,.
 \label{index-22-vector}
\end{equation}
One can show that the index (\ref{index-N2-int}) is invariant under the exchange of fugacities $\ba\leftrightarrow \bb$ or, equivalently, $x\leftrightarrow y$. Therefore we would like to conjecture that, as in the $(0,4)$ and $(4,4)$ cases, the spectrum of the SCFT at the IR fixed point is invariant under the exchange of flavor symmetries $U(1)_x\leftrightarrow U(1)_y$

\subsection{Duality to a $\CN=(0,2)$ Landau-Ginzburg theory}

In the case of $\CN=(0,2)$ one can check that the index (\ref{index-N2-int}) satisfies the following identity:
\begin{equation}
 \CI^{(0,2)}_{\schannel} (\ba, \bb, x, y)=\frac{\theta(x^Ny^N)}{\theta(x^N)\theta(y^N)\prod_{i,j}\theta(xya_ib_j)}
 \label{I02-LG}
\end{equation}
from which the symmetry under the exchange $x\leftrightarrow y$ becomes obvious. This result can be reformulated in the following way. Let us define
\begin{equation}
 \CI^{(0,2)}_{K_N}(\ba,\bb^{-1},x)\equiv\frac{\theta(q/x^N)}{\prod_{i,j}\theta(xa_i/b_j)}\,.
 \label{index-KN}
\end{equation}
which can be understood as the index of the $(0,2)$ Landau-Ginzburg model $K_N^{(0,2)}$ with $N^2$ chiral multiplets $\{\Phi_{i}^j\}_{i,j=1}^N$ with R-charge $0$, Fermi multiplet $\Gamma$ with $R$-charge $1$ and superpotential
\begin{equation}
 W=\Gamma\,\det \Phi.
\end{equation}
The superpotential imposes the condition
\begin{equation}
 \det \Phi=0
 \label{conifold-eq}
\end{equation}
and breaks $U(N^2)$ flavor symmetry of $N^2$ free chirals to $SU(N)_a\times SU(N)_b\times U(1)_x$. The equation (\ref{conifold-eq}) describes a $(N^2-1)$-dimensional conifold $\CC_N$ embedded in $\mathbb{C}^{N^2}$. In particular
\begin{equation}
 \CC_2=\{\Phi_1^1\Phi_2^2-\Phi_1^2\Phi_2^1=0\}
\end{equation}
is the Calabi-Yau threefold usually referenced to as just ``conifold'' in the literature. Then the equation (\ref{I02-LG}) can be written as
\begin{equation}
   \frac{1}{N!}\int\limits_{\textrm{JK}} \frac{d\bxi}{2\pi i\bxi}\,
  \CI^{(0,2)}_{K_N} (\ba, \bxi^{-1}, x) \,
  \CI^{(0,2)}_{V, SU(N)}(\bxi)\,
 \CI^{(0,2)}_{K_N} (\bxi,\bb^{-1}, 1/y) \,
 = \CI^{(0,2)}_{K_N} (\ba, \bb^{-1}, x/y)
   \label{index-KN-int}
\end{equation}
Physically (\ref{index-KN-int}) means that gauging a diagonal subgroup of $SU(N)\times SU(N)$ flavor symmetry from two copies of $K^{(0,2)}_N$ is dual to just one copy of $K^{(0,2)}_N$. Let $(\Phi^{(1)})_i^\alpha$, $(\Phi^{(2)})_\beta^j$ be chiral fields from two copies of $K^{(0,2)}_N$ in the l.h.s. of duality. The conditions $\det \Phi^{(1,2)}=0$ kill baryons of the theory in the chiral ring. This means that we are only left with mesons $\Phi^i_j\equiv(\Phi^{(1)})^\alpha_j (\Phi^{(2)})_\alpha^i$ which play the roles of chiral fields of the dual Landau-Ginzburg model. The condition $\det\Phi =0$ is obviously satisfied and one can also show there are no additional conditions on $\Phi$. Geometrically the statement can be understood as the following relation:
\begin{equation}
 (\CC_N\times \CC_N )\kquot SU(N) \cong \CC_N.
\end{equation}
Also, this duality is similar to a $\mathcal{N}=(0,2)$ Seiberg-like duality found in \cite{Gadde:2013lxa} in the case when there are no Fermi multiplets in fundamental representation of the gauge group. There is an important difference however, theories considered in the aforementioned paper had $U(N)$ gauge symmetry, not $SU(N)$.

As we show in appendix \ref{appendixinversion}, the identity (\ref{index-KN-int}) can be used to derive an iversion formula for a certain integral operator with kernel constructed from theta-functions. It is analogous to the inversion formula  in \cite{gamma-inversion-math} for an operator with kernel constructed in a similar way from elliptic Gamma functions and allows us to find an explicit expression for the index of $T_3^{(0,4)}$ theory in section \ref{section:T3}.

%%%%%%%%%%%%%%%%%%%%%%%%%%%%%%%%%%%%%%%%%%%%%%%%

\acknowledgments{We would like to thank J.~Andersen, I.~Bah, A.~Gadde, S.~Gukov, K. Intriligator, S.~Kachru, Sj.~Lee, D.~Nemeschansky, D.~Pei, V.~Stylianou, D.~Xie for useful discussions.
The work of P.P. is supported in part by the Sherman Fairchild scholarship and by NSF Grant PHY-1050729.
The work of J.S. is supported by the US Department of Energy under UCSD's contract de-sc0009919.
The work of W.Y. is supported in part by the Sherman Fairchild scholarship and by DOE grant DE-FG02-92-ER40701.
P.P. would like to thank Centre for Quantum Geometry of Moduli Spaces in Aarhus University for hospitality during the final stage of the work.
J.S. would like to thank Enrico Fermi Institute at the University of Chicago for hospitality during the final stage of the work.
Opinions and conclusions expressed here are those of the authors and do not necessarily reflect the views of funding agencies.}

\appendix
\section{Review on $\CN=(0, 2)$ and $\CN=(0, 4)$ theory}
\label{appendix04}

Let us summarize some basic facts about $\CN=(0, 2)$ and $\CN=(0, 4)$ gauge theories \cite{Witten:1993yc}. See also \cite{Edalati:2007vk,Tong:2014yna}.

\paragraph{$\CN=(0, 2)$ multiplets}
A general $\CN=(0, 2)$ gauge theory can have the following supersymmetry multiplets:
\be
\centering
\begin{tabular}{|c|c|c|}
\hline
 Multiplets & Superfield & Components (on-shell) \\
 \hline
 Vector &$U$  & $(A_\mu, \lambda_-)$ \\
 Chiral & $\Phi$ & $(\psi_+, \phi)$ \\
 Fermi & $\Psi$ & $(\psi_-)$ \\
 \hline
\end{tabular}
\ee
Here, the subscript $\pm$ stands for right/left-moving complex Weyl spinors respectively. An  $\mathcal{N}=(0,2)$ theory allows formulation in $(x^\pm,\theta^+,\bar\theta^+)$ superspace. A chiral superfield satisfies
\be
 \bar{\CD}_+ \Phi = 0 \ ,
\ee
and has the following expansion:
\begin{equation}
 \Phi=\phi+\sqrt{2}\theta^{+}\psi_{+}-i\theta^{+}\bar{\theta}^{+}\partial_{+}\phi.
\end{equation}

A Fermi superfield satisfies
\be \label{eq:Fermi}
 \bar{\CD}_+ \Psi = E (\Phi_i) \ ,
\ee
where $E(\Phi_i)$ is a holomorphic function of the chiral superfields $\Phi_i$ which transforms in the same way as $\Psi$. This condition leads to the following expansion:
\begin{equation}
 \Psi=\psi_{-}-\sqrt{2}\theta^{+}G-i\theta^{+}\bar{\theta}^{+}\partial_{+}\psi_{-}-\sqrt{2}\bar{\theta}^{+}E.
\end{equation}
where $G$ is an auxillary superfield. Finally, the vector superfield has the following form:
\begin{equation}
 U=A_--2i\theta^{+}\lambda_{-}-2i\bar{\theta}^{+}\bar{\lambda}_{-}+2\theta^{+}\bar{\theta}^{+}D.
\end{equation}
The corresponding field strength forms a Fermi superfield $\Lambda$, which is consistent with the fact that (bosonic) vector field in 2d is non-dynamical.

There are two different types of `superpotential' in $\CN=(0, 2)$ theory. To each Fermi multiplets $\Psi_a$, introduce a holomorphic function $J^a
(\Phi_i)$. Then we write the SUSY action
\be
 S_J = \int d^2 x d\theta^+  \Psi_a J^a (\Phi_i) + \textrm{h.c} \ .
\ee
We can write `superpotential' as $W = \Psi_a J^a(\Phi)$, and integrate over the half-superspace.

There is also $E$-type superpotential, which appears in the right-hand side of the \eqref{eq:Fermi}. There is one condition we need to impose to ensure
supersymmetry:
\be
 E \cdot J \equiv \sum_a E_a J^a = 0 \ .
\ee

\paragraph{$\CN=(0, 4)$ multiplets}
There is no simple superspace formalism in the case of $\CN=(0, 4)$ supersymmetry. An $\CN=(0,4)$ gauge theory is usually formulated in terms of combinations of $\CN=(0,2)$ which combine into the following $\CN=(0,4)$ multiplets:
\be
\centering
\begin{tabular}{|c|c|c|c|}
	\hline
	Multiplets & $\CN=(0, 2)$ superfields & Components & $SU(2)_R^- \times SU(2)_R^+$ \\
	\hline
	Vector & vector + Fermi $(U, \Theta)$ & $(A_\mu, \lambda_-^{a})$ & $(1, 1), (2, 2)$ \\
	Hypermultiplet & chiral + chiral $(\Phi, \tilde{\Phi})$ & $(\phi^a , \psi_{+, b})$ & $(2, 1), (1, 2)$ \\
	Twisted hyper & chiral + chiral $(\Phi', \tilde{\Phi}')$ & $(\phi'_a, \psi'^{b}_+ )$ & $(1, 2), (2, 1)$ \\
	Fermi & Fermi + Fermi $(\Gamma, \tilde{\Gamma})$ & $(\psi_-^a)$ & $(1, 1)$\\
	\hline
\end{tabular}
\ee
Here $a, b=1, 2$. We remark that $\CN=(0, 4)$ supersymmetry in principle does not require $\CN=(0, 4)$ Fermi multiplets to have two copies of $\CN=(0, 2)$ Fermi multiplets (see e.g. \cite{Witten:1994tz}). In our case, as in \cite{Tong:2014yna},  we define a (0,4) Fermi multiplet as a pair of Fermi multiplets in the conjugate representations.

When a hypermultiplet couples to a vector multiplet, we have a superpotential coupling between the hypermultiplet and Fermi multiplet $\Theta$ in the vector given as
\be
 J^\Theta = \Phi \tilde{\Phi} \ , \qquad  \quad W = \tilde{\Phi}\Theta \Phi \ .
\ee
This is analogous to the superpotential coupling in 4d $\CN=2$ theory between chiral adjoint in a vector multiplet and a hypermultiplet.

For a twisted hypermultiplet, the coupling is done through the $E$-term, instead of the superpotential (or $J$-term). It is given by
\be
E_\Theta = \Phi' \tilde{\Phi}' \ ,
\ee
where the right-hand side of the equation transform as the adjoint of the gauge group.

For the case of Fermi multiplet, there is no coupling between $\Theta$ and $\Gamma, \tilde{\Gamma}$. But, it is possible to include a quadratic $E$ or $J$ term while
preserving the $SO(4)_R$ symmetry.

\paragraph{$\CN=(4, 4)$ multiplets}
can be understood as pairs of $\CN=(0,4)$ multiplets:
\be
\centering
\begin{tabular}{|c|c|c|}
	\hline
	Multiplets &$\CN=(0, 4)$ multiplets & $\CN=(0, 2)$ superfields \\
	\hline
	Vector & vector + twisted hyper & $(U, \Theta), (\S, \tilde{\S})$  \\
	Hypermultiplet & hyper + Fermi & $(\Phi, \tilde{\Phi}), (\Gamma, \tilde{\Gamma})$\\
	\hline
\end{tabular}
\ee
An $\CN=(4, 4)$ vector multiplet contains adjoint valued twisted hypermultiplet. The $\CN=(0, 2)$ chiral multiplets in the twisted hypermultiplet couple
with the $\CN=(0, 4)$ vector multiplet via
\be
 E_\Theta = [\S, \tilde{\S}] \ .
\ee
And a hypermultiplet couples with vector multiplet with
\be
 W = \tilde{\Phi} \Theta \Phi \ .
\ee
There is also a coupling between $\CN=(0, 4)$ Fermi, hyper and a twisted hypermultiplet. It involves $J$-term given as
\be
 W = \tilde{\Gamma} \tilde{\S} \Phi + \tilde{\Phi} \tilde{\S} \G \ ,
\ee
and also the $E$-term
\be
 E_\G = \S \Phi \ ,  \quad E_{\tilde{\G}} = -\tilde{\Phi} \S \ .
\ee
These terms satisfy the constraint $E \cdot J = 0$.

One can obtain $\CN=(4, 4)$ multiplets starting from 6d $\CN=(1, 0)$ gauge theories and then dimensionally reducing to 2d. In 6d, we have $SU(2)_R$
symmetry. The vector inside a vector multiplet is a singlet under the $SU(2)_R$. A hypermultiplet contains complex scalars in the doublet of $SU(2)_R$.
Upon dimensional reduction, we get R-symmetry $SO(4)_R = SU(2)_l \times SU(2)_r$. The left/right-moving supercharges are in $(2, 1, 2)/(1, 2, 2)$
representations of $SU(2)_l \times SU(2)_r \times SU(2)_R$. The charges of the component fields are as follows:
\be
\centering
\begin{tabular}{|c|c|c|}
	\hline
	Multiplets & components & $SU(2)_l \times SU(2)_r \times SU(2)_R$ \\
	\hline
	Vector & $A_\mu$ & $ (1, 1, 1)$ \\
		& $\phi$ & $(2, 2, 1)$ \\
		& $\lambda_-$ & $(1, 2, 2)$ \\
		& $\lambda_+$ & $(2, 1, 2)$ \\
	\hline
	Hypermultiplet & $q$ & $(1, 1, 2)$ \\
		& $\psi_-$ & $(2, 1, 1)$ \\
		& $\psi_+$ & $(1, 2, 1)$ \\
	\hline
\end{tabular}
\ee
Here $SU(2)_R = SU(2)_R^-$, $SU(2)_r = SU(2)_R^+$ and $SU(2)_l = SU(2)_I$. The other R-symmetry $SU(2)_l$ becomes the global symmetry for $(0, 4)$ theories.

Note that the scalar in the hypermultiplet is uncharged under $SU(2)_l \times SU(2)_r$ but charged under $SU(2)_R$, whereas the scalar in the vector
multiplet is charged under the $SU(2)_R$ but uncharged under $SU(2)_l \times SU(2)_r$. It has been argued that $\CN=(4, 4)$ gauge theory flows to two
decoupled SCFTs on the Higgs branch and the Coulomb branch \cite{Witten:1997yu, Aharony:1999dw}. For a large value of these scalar fields, we can trust
the semi-classical description, which is given by the Higgs/Coulomb branch. For the Higgs branch theories, the R-symmetry should be given by $SU(2)_l
\times SU(2)_r$ since the scalars are charged under $SU(2)_R$. It is the other way around for the Coulomb branch theories. (Here the extra $SU(2)$ R-symmetry is not visible in the UV.) Since R-symmetries on the Coulomb branch and Higgs branch are distinct, they cannot be the same SCFT.

%%%%%%%%%%%%%%%%%%%%%%%%%%%%%%%%%%
\section{Review on elliptic genus}
\label{appendixindex}

\subsubsection*{Elliptic genus for $(0,2)$ gauge theories}

The elliptic genus of $\CN=(0,2)$ supersymmetric theories was discussed in \cite{Gadde:2013wq,Benini:2013nda,Benini:2013xpa}. We will summarize the prescription for computing the elliptic genus of $\CN=(0,2)$ theories in this section.

Consider a two-dimensional theory with $\CN=(0,2)$ supersymmetry and a flavor symmetry group $\CF$. The elliptic genus on Ramond (R) sector is defined as
\begin{equation}
  \CI^{(0,2),R}(\ba;q)=\Tr_{R}(-1)^Fq^{H_L}\bar{q}^{H_R}\prod_{i}a_i^{f_i},
\end{equation}
while the elliptic genus on Neveu-Schwarz (NS) sector is defined as
\begin{equation}
  \CI^{(0,2),NS}(\ba;q)=\Tr_{NS}(-1)^Fq^{H_L}\bar{q}^{H_R-\half J_R}\prod_ia_i^{f_i},
\end{equation}
where $\Tr_{R}$ or $\Tr_{NS}$ are taken over the Hilbert space of SCFT on a circle, with fermions satisfying periodic or anti-periodic boundary conditions respectively. $F$ is the fermion number, and the parameter
\begin{equation}q=e^{2\pi i\tau}\end{equation}
specifies the complex structure of a torus. $H_L$ is the left-moving Hamiltonian, $H_R$ and $J_R$ are the right-moving Hamiltonian and $U(1)_R$ charge operator, $f_i$'s are the Cartan generators of $\CF$, and $a_i$ are corresponding fugacities. The collection of fugacities $\ba\equiv \{a_i\}$ can be understood as the element of the maximal torus of $\CF$. By the usual argument both elliptic genera are independent of $\bar{q}$.

The contribution of a chiral multiplet $\Phi$ transforming in a representation $\CR$ is
\begin{equation}
  \CI^{(0,2),R}_{\Phi,\CR}(\bx;q)=\prod_{\rho\in\CR}\frac{1}{\tth(\bx^\rho;q)},\quad
  \CI^{(0,2),NS}_{\Phi,\CR}(\bx;q)=\prod_{\rho\in\CR}\frac{1}{\theta(q^{\frac{r}{2}}\bx^\rho;q)}.
\end{equation}
Where whe product is over the weights of $\rho$ of the representation $\CR$, and $\bx^\rho\equiv \prod _i x_i^{\langle f_i,\rho\rangle}$ denotes the standard pairing between an element of the maximal torus and a weight. The contribution of a Fermi multiplet $\Psi$ in a representation $\CR$ is
\begin{equation}
  \CI^{(0,2),R}_{\Psi,\CR}(\bx;q)=\prod_{\rho\in\CR}(-\tth(\bx^\rho;q)),\quad
  \CI^{(0,2),NS}_{\Psi,\CR}(\bx;q)=\prod_{\rho\in\CR}{\theta(q^{\frac{r+1}{2}}\bx^\rho;q)}.
\end{equation}
The theta function is defined as
\begin{equation}
  \theta(x;q)=(x;q)(q/x;q), \quad\quad \tth(x;q)=x^{-\half}\theta(x;q),
\end{equation}
where
\begin{equation}
 (x;q)=\prod_{i=0}^{\infty}(1-xq^i).
\end{equation}
Notice that the NS-NS elliptical genera for chiral and Fermi multiplet depend on the right-moving $J_R$-charge $r$ of the multiplet.

The contribution of a vector multiplet $\Lambda$ with gauge group $G$ is
\begin{equation}
\begin{split}
  &\CI^{(0,2),R}_{\Lambda,G}(\bz;q)=(q;q)^{2\,\mathrm{rk}\,G}\prod_{\substack{\alpha\in \text{adj}_G \\ \alpha\neq 0 }}(-\tth(\bz^\alpha;q)),\\
  &\CI^{(0,2),NS}_{\Lambda,G}(\bz;q)=(q;q)^{2\,\mathrm{rk}\,G}\prod_{\substack{\alpha\in \text{adj}_G \\ \alpha\neq 0 }}\theta(\bz^\alpha;q).
\end{split}
\end{equation}
Here $\mathrm{rk}\,G$ is the rank of gauge group $G$ and $\bz$ is the element of the maximal torus of the gauge group $G$.

The elliptic genus does not depend on the coupling of the theory, therefore it is always possible to compute it in the free theory limit. For a $(0,2)$ gauge theory with gauge group $G$, chiral multiplets $\{\Phi\}$ and Fermi multiplets $\{\Psi\}$, the elliptic genus of the theory is \cite{Gadde:2013wq,Benini:2013nda,Gadde:2013dda,Benini:2013xpa}:
\begin{multline}
  \CI^{(0,2),R|NS}(\ba;q)=
  \frac{1}{W(G)}\int\limits_\text{JK}\prod_{i=1}^{\mathrm{rk}\,G}\frac{d z_i}{2\pi i z_i}
             \CI^{(0,2),R|NS}_{V,G}(\bz;q)\times \\
             \prod_{\Phi}\CI^{(0,2),R|NS}_{\Phi}(\{\ba,\bz\};q)
    \prod_{\Psi}\CI^{(0,2),R|NS}_{\Psi}(\{\ba,\bz\};q)
\end{multline}
where $W(G)$ is the order of Weyl group of $G$. The integral is performed over a certain contour ``JK'' in the moduli space of flat connections on the two-torus $\CM_\text{flat}(T^2_\tau,G)$ which corresponds to taking a sum of Jeffrey-Kirwan residues. The absence of gauge anomaly is equivalent to the condition that the integrand is elliptic in $\bz$.

\subsubsection*{Elliptic genus for $\CN=(0,4)$ theory}

To compute the elliptic genus for two-dimensional theories with $(0,4)$ supersymmetry, one can first decompose the $(0,4)$ supersymmetric algebra into its $(0,2)$ subalgebra. The $R$-symmetry of $(0,4)$ is $SU(2)^-_R\times SU(2)^+_R$ from which the combination $J_R=(1-\alpha)R^-+(1+\alpha)R^+$ is chosen as $(0,2)$ $R$-charge. The other combination $R_v=2(R^--R^+)$ can be treated as a global symmetry in $(0,2)$ algebra.

With the embedding of $(0,2)$ algebra into $(0,4)$ algebra and the decomposition of $(0,4)$ multiplets discussed in appendix \ref{appendix04}, one can write down the elliptic genus for $(0,4)$ multiplets. For half-hyper multiplets we have
\begin{equation}
  \CI^{(0,4),R}_{\Phi,\CR}(\bx;q)=\prod_{\rho\in\CR}\frac{1}{\tth(v\bx^\rho;q)},\quad
  \CI^{(0,4),NS}_{\Phi,\CR}(\bx;q)=\prod_{\rho\in\CR}\frac{1}{\theta(q^{\frac{1-\alpha}{4}}v\bx^\rho;q)},
\end{equation}
where the fugacity $v$ labels the anti-diagonal Cartan $F$ of $SU(2)^-_R\times SU(2)^+_R$ mentioned above. For half twisted-hyper,
\begin{equation}
  \CI^{(0,4),R}_{\Phi',\CR}(\bx;q)=\prod_{\rho\in\CR}\frac{1}{\tth(v^{-1}\bx^\rho;q)},\quad
  \CI^{(0,4),NS}_{\Phi',\CR}(\bx;q)=\prod_{\rho\in\CR}\frac{1}{\theta(q^{\frac{1+\alpha}{4}}v^{-1}\bx^\rho;q)}.
\end{equation}
The elliptic genus of $(0,4)$ Fermi multiplet is
\begin{equation}
  \CI^{(0,4),R}_{\Psi,\CR}(\bx;q)=\prod_{\rho\in\CR}(-\theta(\bx^\rho;q)),\quad
  \CI^{(0,4),NS}_{\Psi,\CR}(\bx;q)=\prod_{\rho\in\CR}{\theta(q^{\frac{1}{2}}\bx^\rho;q)}.
\end{equation}
And finally the vector multiplet,
\begin{equation}
\begin{split}
  &\CI^{(0,4),R}_{\Lambda,G}(\bz;q)=(\tth(v^{-2};q))^{\mathrm{rk}\,G}\prod_{\substack{\alpha\in \text{adj}_G \\ \alpha\neq 0 }}\tth(v^{-2}\bz^\alpha;q)\tth(\bz^\alpha;q),\\
  &\CI^{(0,4),NS}_{\Lambda,G}(\bz;q)=(\theta(q^{\frac{1+\alpha}{2}}v^{-2};q))^{\mathrm{rk}\,G}\prod_{\substack{\alpha\in \text{adj}_G \\ \alpha\neq 0 }}\theta(q^{\frac{1+\alpha}{2}}v^{-2}\bz^\alpha;q)\theta(\bz^\alpha;q).
\end{split}
\end{equation}
Notice that in the main text we simply choose $\alpha=1$.

\subsubsection*{Elliptic genus for $\CN=(2,2)$ theory}

In $(2,2)$ theory there are chiral and vector multiplets. $(2,2)$ chiral multiplet decomposes into a $(0,2)$ chiral and a $(0,2)$ Fermi, while a $(2,2)$ vector multiplet is composed of a $(0,2)$ vector and a $(0,2)$ Fermi, therefore one can write down the elliptic genus for $(2,2)$ theory accordingly. Here we just summarize the results, details can be found in \cite{Benini:2013nda,Gadde:2013dda,Benini:2013xpa}.
\begin{equation}
  \CI^{(2,2),R}_{\Phi,\CR}(\bx;q)=\prod_{\rho\in\CR}\frac{\tth(y^{R/2-1}\bx^\rho;q)}{\tth(y^{R/2}\bx^\rho;q)},\quad
  \CI^{(2,2),NS}_{\Phi,\CR}(\bx;q)=\prod_{\rho\in\CR}\frac{\theta(q^{\half(R/2+1)} y^{R/2-1}\bx^\rho;q)}{\theta(q^{R/4} y^{R/2}\bx^\rho;q)},
\end{equation}
where the fugacity $v$ labels the anti-diagonal Cartan $F$ of $SU(2)^-_R\times SU(2)^+_R$ mentioned above. And the vector multiplet,
\begin{equation}
\begin{split}
  &\CI^{(2,2),R}_{\Lambda,G}(\bz;q)=\left(\frac{(q;q)^2}{\tth(y^{-1};q)}\right)^{\mathrm{rk}\,G}\prod_{\substack{\alpha\in \text{adj}_G \\ \alpha\neq 0 }}\frac{\tth(\bz^\alpha;q)}{\tth(y^{-1}\bz^\alpha;q)},\\
  \\
  &\CI^{(2,2),NS}_{\Lambda,G}(\bz;q)=\left(\frac{(q;q)^2}{\theta(q^\half y^{-1};q)}\right)^{\mathrm{rk}\,G}\prod_{\substack{\alpha\in \text{adj}_G \\ \alpha\neq 0 }}\frac{\theta(\bz^\alpha;q)}{\theta(q^\half y^{-1}\bz^\alpha;q)}.
\end{split}
\end{equation}
In NS-NS index we sometimes use a new fugacity $t=q^\half/y$ instead of $y$.

\subsubsection*{Elliptic genus for $\CN=(4,4)$ theory}

In $(4,4)$ theory there are also hyper multiplets and vector multiplets like $(0,4)$ cases. The single letter indices for half hyper multiplets are
\begin{equation}
  \CI^{(4,4),R}_{\Phi,\CR}(\bx;q)=\prod_{\rho\in\CR}\frac{\tilde{\theta}(u\bx^\rho;q)}{\tilde{\theta}(v\bx^\rho;q)},\quad
  \CI^{(4,4),NS}_{\Phi,\CR}(\bx;q)=\prod_{\rho\in\CR}\frac{\theta(u\bx^\rho;q)}{\theta(v\bx^\rho;q)},
\end{equation}
the single letter indices for vector multiplets are
\begin{equation}
\begin{split}
  \CI^{(4,4),R}_{\Lambda,G}(\bz;q)=&\left(\frac{\tth(v^{-2};q)}{\tth(uv^{-1};q)\tth(u^{-1}v^{-1};q)}\right)^{\mathrm{rk}\,G}\prod_{\substack{\alpha\in \text{adj}_G \\ \alpha\neq 0 }}\frac{\tth(v^{-2}\bz^\alpha;q)\tth(\bz^\alpha;q)}{\tth(uv^{-1}\bz^\alpha;q)\tth(u^{-1}v^{-1}\bz^\alpha;q)},\\
  \\
  \CI^{(4,4),NS}_{\Lambda,G}(\bz;q)=&\left(\frac{\theta(qv^{-2};q)}{\theta(q^{\half}uv^{-1};q)\theta(q^{\half}u^{-1}v^{-1};q)}\right)^{\mathrm{rk}\,G}\prod_{\substack{\alpha\in \text{adj}_G \\ \alpha\neq 0 }}\frac{\theta(qv^{-2}\bz^\alpha;q)\theta(\bz^\alpha;q)}{\theta(q^{\half}uv^{-1}\bz^\alpha;q)\theta(q^{\half}u^{-1}v^{-1}\bz^\alpha;q)}.
\end{split}
\end{equation}

\section{'t Hooft anomalies}

\label{appendixanomalies}

In theories with chiral supersymmetry left- and right-moving fermions are not necessarily paired together, which in general results in non-trivial 't Hooft anomalies. Suppose the theory under consideration has a global symmetry with corresponding simple Lie group $F$. Then its anomaly coefficient $k_F$ is given by the following formula:
\begin{equation}
 \Tr\gamma^3F^aF^b =k_F\delta^{ab},
\end{equation}
where $F^a$ are the generators of $F$, $\gamma^3$ is the gamma matrix measuring chirality and the trace is performed over the space of Weyl Fermi fields of the theory. It follows that the anomaly coefficient $k_F$ can be calculated as the following difference between sums over the sets of (0,2) chiral and Fermi multiplets of the theory:
\begin{equation}
 k_F=\sum\limits_{\Phi\in\text{(0,2) chiral}}T(R^\Phi_{F})-\sum\limits_{\Gamma\in\text{(0,2) Fermi}}T(R^\Gamma_{F}),
\end{equation}
where $T(R^{\,\cdot}_F)$ denotes the index of representation $R^{\,\cdot}_F$ of $F$. For example, $T[\square_{SU(N)}]=1/2$ and $T[\text{adj}_{SU(N)}]=N$. In the case when the theory has two $U(1)$ symmetries $U(1)_{F_{1,2}}$ with corresponding charges $F_{1,2}$, there can be a mixed 't Hooft anomaly:
\begin{equation}
 k_{F_1\cdot F_2}=\Tr\gamma^3F_1F_2.
\end{equation}
However, unlike in 4d there cannot be a mixed anomaly between $SU(N)$ and other global symmetry.

In the IR one usually expects the current corresponding to the global symmetry to become holomorphic or anti-holomorphic (i.e. left- or right-moving). In this case $F$ enhances to the corresponding affine algebra $\widehat{F}_{|2k_F|}$ acting in the holomorphic or anti-holomorphic sector of the CFT depending on the sign of $k_F$. However, holomorphicity of the current in the IR may fail if the flavor symmetry rotates non-compact directions of the moduli space, the simplest example being $U(1)$ symmetry acting on a free chiral multiplet.

The anomaly coefficient determines transformation properties of the index w.r.t. to corresponding fugacities. The index can be considered as a meromorphic section of $\CL^{-2k_F}$ where $\CL$ is a prequantum line bundle over $\CM_{\text{flat}}(T^2_{\tau},F)$, the moduli space of flat connections of $F$-bundle over the two-torus with complex structure $\tau$. Consider for example the case $F=SU(n)$. Let us denote the corresponding fugacities by $\ba=\{a_i\}_{i=1}^N$, $\prod_i a_i=1$. Then the index has the following properties:
\begin{equation}
I(\ba|_{a_i\leftrightarrow a_j})=I(\ba),
\qquad I(\ba|_{a_i\rightarrow qa_i,a_j\rightarrow a_j/q})=(qa_i/a_j)^{2k_F}I(\ba).
\end{equation}

Since $\CN=2$ or small $\CN=4$ SCA algebra of the IR SCFT has only one central element, the anomaly of the R-symmetry can be related to the the right-moving central charge. Namely, in the case of $\CN=2$ SCA:
\begin{equation}
 c_R=3k=3\Tr \gamma^3R^2,
\end{equation}
where $R$ is the generator of $U(1)$ R-symmetry and $k$ is the level of affine $\widehat{U(1)}$ R-symmetry. In the case of small $\CN=4$ SCA:
\begin{equation}
 c_R=6k=6\cdot (2k_R),
\end{equation}
where $k$ is the level of affine $\widehat{SU(2)}$ R-symmetry and $k_R$ is the corresponding anomaly coefficient which usually can be easily computed in the UV. Once $c_R$ is known the left-moving central charge can be easily determined from the gravitational anomaly:
\begin{equation}
 c_L-c_R=\Tr\gamma^3.
\end{equation}

\section{Proof of the elliptic inversion formula}
\label{appendixinversion}

\begin{defn}
 Let $\CH^{(m)}_{SU(2)}$ be the space of meromorphic sections with simple poles\footnote{We make this assumption for technical simplicity. The case with higher order poles can always be considered as a limit when simple poles collide.} on $\CL^{-m}$ where $\CL$ is the prequantum line bundle on $\CM_\mathrm{flat}(T^2_\tau,SU(2))\cong T^2_\tau/\mathbb{Z}_2$. More explicitly\footnote{cf appendix \ref{appendixanomalies}},
 \begin{equation}
   \CH_{SU(2)}^{(m)}\equiv \{ f:\mathbb{C}^* \rightarrow \mathbb{C}\,|\, f(z)=f(1/z),\,f(qz)=q^{m}z^{2m}f(z)\}.
 \end{equation}
\end{defn}

\begin{prop}
 If $f\in \CH^{(m)}_{SU(2)},\;m>0$ has no poles, it is zero.
 \label{prop-1}
\end{prop}

\begin{proof}
 Consider $\tilde{f}(z)=f(z)(\theta(z)\theta(1/z))^m$. It is an elliptic function without poles, therefore it must be constant: $\tilde{f}(z)\equiv C$. Since $f(z)$ has no poles $C=0$\footnote{In other words, $f$ is a section of a line bundle over $\CM_\mathrm{flat}(T^2_\tau,SU(2))$ with divisor $-m\cdot \mathrm{pt}$ and therefore it must have at least $m$ poles}.
\end{proof}

It follows that in order to prove the equality of two functions with positive anomaly coefficients and simple poles it is sufficient to check that they have the same poles and residues. In particular, it is easy to show that
\begin{prop}
\label{prop-decomp}
 If $f\in \CH^{(1)}_{SU(2)}$, $\exists \,A_i,t_i$ (unique up to a $\mathbb{Z}_2$ action) such that
 \begin{equation}
  f(z)=\sum_{i}\frac{A_i}{\theta(t_iz)\theta(t_i/z)}.
  \label{prop-decomp-expr}
 \end{equation}
\end{prop}

\begin{lem}
\begin{equation}
  \frac{(q;q)^2}{2}\int\limits_\mathrm{JK}\frac{d\xi}{2\pi i\xi}\,\theta(\xi^{\pm 2})
  \,\frac{\theta(x^2)}{\theta(x\,a^{\pm 1}\xi^{\pm 1})}
  \,\frac{\theta(y^2)}{\theta(y\,\xi^{\pm 1}b^{\pm 1})}=
  \,\frac{\theta(x^2y^2)}{\theta(xy\,a^{\pm 1}b^{\pm 1})}
  \label{inv-basic}
\end{equation}
\end{lem}

\begin{proof}
By definition the integral on left hand side is given by a  residues at $\xi=xa^{\pm 1}$ and $\xi=yb^{\pm 1}$:
 \begin{multline}
\frac{\theta \left(y^2\right) \theta \left(\frac{a^2}{x^2}\right)}{2 \theta \left(a^2\right) \theta \left(\frac{a y}{b x}\right) \theta \left(\frac{a b y}{x}\right) \theta \left(\frac{x y}{a b}\right) \theta \left(\frac{b x y}{a}\right)}+\frac{\theta \left(y^2\right) \theta \left(\frac{1}{a^2 x^2}\right)}{2 \theta \left(\frac{1}{a^2}\right) \theta \left(\frac{y}{a b x}\right) \theta \left(\frac{b y}{a x}\right) \theta \left(\frac{a x y}{b}\right) \theta (a b x y)}+
\\
\frac{\theta \left(x^2\right) \theta \left(\frac{b^2}{y^2}\right)}{2 \theta \left(b^2\right) \theta \left(\frac{b x}{a y}\right) \theta \left(\frac{a b x}{y}\right) \theta \left(\frac{x y}{a b}\right) \theta \left(\frac{a x y}{b}\right)}+\frac{\theta \left(x^2\right) \theta \left(\frac{1}{b^2 y^2}\right)}{2 \theta \left(\frac{1}{b^2}\right) \theta \left(\frac{x}{a b y}\right) \theta \left(\frac{a x}{b y}\right) \theta \left(\frac{b x y}{a}\right) \theta (a b x y)}.
\label{basic-int-expression}
 \end{multline}
It is easy to show that, as a function of $a$ which belongs to $\CH^{(2)}_{SU(2)}$, it has the same poles and residues as the right hand side of (\ref{inv-basic}). By Prop. \ref{prop-1} the difference between (\ref{basic-int-expression}) and the right hand side of (\ref{inv-basic}) is zero.
\end{proof}

The formula (\ref{inv-basic}) is a particular case of (\ref{index-KN-int}) for $N=2$. Now it is easy to prove the following statement:

 \begin{thm}
 For any $f\in \CH^{(2)}_{SU(2)}$
\begin{equation}
  \frac{(q;q)^4}{4}\int\limits_\mathrm{JK}\frac{d\xi}{2\pi i\xi}\,\int\limits_\mathrm{JK}\frac{d\zeta}{2\pi i\zeta}\,\theta(\xi^{\pm 2})\,\theta(\zeta^{\pm 2})
  \,\frac{\theta(v^{-2})}{\theta(v^{-1}\,z^{\pm 1}\xi^{\pm 1})}
  \,\frac{\theta(v^{2})}{\theta(v\,\xi^{\pm 1}\zeta^{\pm 1})}f(\zeta)=
  f(z)
  \label{inv-formula}
\end{equation}
\end{thm}

\begin{proof}
Let us pick some $a\in\mathbb{C}^*$ and consider
\begin{equation}
 \tilde{f}(z)=\theta(az)\theta(a/z)f(z)\;\in\;\CH^{(1)}_{SU(2)}.
\end{equation}
Then from Prop. \ref{prop-decomp} it follows that we can always represent $f$ in the following way\footnote{Let us note that the Jeffrey-Kirwan contour integral prescription in (\ref{inv-formula}) requires the choice of $SU(2)$ charges at poles. This choice is made in the formula below by picking particular $(A_i,t_i)$ in $\mathbb{Z}_2$ orbit when using representation (\ref{prop-decomp-expr}). However, the final result obviously does not depend on it.}:
\begin{equation}
 f(z)=\sum_{i}\frac{A_i}{\theta(az)\theta(a/z)\theta(t_iz)\theta(t_i/z)}.
\end{equation}
Plugging it in the left hand side of (\ref{inv-formula}) and applying (\ref{inv-basic}) twice for each term in the sum we get the desired result.
\end{proof}
Let us note that one can easily generalize the above statements for $SU(N)$ case, considering the following space:
\begin{equation}
  \CH_{SU(N)}^{(m)}\equiv \left\{\text{meromorphic sections of } \CL^{-m}\rightarrow \CM_\text{flat}(T^2_\tau,SU(N))\right\}
\end{equation}
and utilizing the identity (\ref{index-KN-int}) for general $N$.

\section{Index of $SU(N)$ $\CN=(0,2)$ gauge theories and 1d TQFT}

Making a simplified analogy with section \ref{section-TQFT}, one can construct a 1d TQFT using (\ref{index-02-vector}) and (\ref{index-KN}). Namely, let us define the Hilbert space associated to a point as a space of meromorphic functions of $SU(N)\times U(1)$ fugacities with fixed $SU(N)$ anomaly coefficient:
\begin{multline}
 \CH_{\pt}^{(0,2)}\equiv
 \{ f: (\mathbb{C}^*)^{N-1}\times \mathbb{C}^* \rightarrow \mathbb{C}\,:\\ \, f(\ba|_{a_i\leftrightarrow a_j};x)=f(\ba;x),\,f(\ba|_{a_i\rightarrow qa_i,a_j\rightarrow a_j/q};x)=(qa_i/a_j)^Nf(\ba;x)\}.
 \label{TQFT1-Hilbert-space}
\end{multline}
Then define the following basic building blocks of 1d TQFT:

\begin{equation}
 \begin{tabular}{|c|c|}
 \hline
$\vcenter{\hbox{\includegraphics[scale=0.7,angle=180]{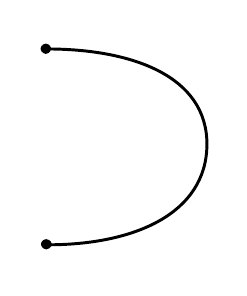}}}$ &
$
\begin{array}{rrcl}
  K:&\, \mathbb{C}& \longrightarrow & \CH_\pt^{(0,2)}\otimes \CH_\pt^{(0,2)} \\
  \\
 & 1 & \longmapsto&\CI^{(0,2)}_{K_N}(\ba,\bb,x\cdot y)
 \end{array}
$
\\
\hline
$\vcenter{\hbox{\includegraphics[scale=0.7]{TQFT1-2map}}}$ &
$
\begin{array}{rrcl}
  \eta:&\,\CH_\pt^{(0,2)}\otimes \CH_\pt^{(0,2)} &\longrightarrow & \mathbb{C}\\
  \\
 & f(\ba,\bb;x,y) & \longmapsto& \frac{1}{N!}\int\limits_\text{JK}\frac{d\bxi}{2\pi i\bxi}\,\CI^{(0,2)}_{V,SU(N)}(\bxi)
 \,f(\bxi,\bxi^{-1};1,1)
 \end{array}
$\\
\hline
\end{tabular}
\label{TQFT1-blocks}
\end{equation}
Again, the last condition in (\ref{TQFT1-Hilbert-space}) is needed for the integrand above to be elliptic. Then (\ref{index-KN-int}) can be formulated as the following property:

\begin{equation}
 \begin{tabular}{|c|c|}
 \hline
$\vcenter{\hbox{\includegraphics[scale=1]{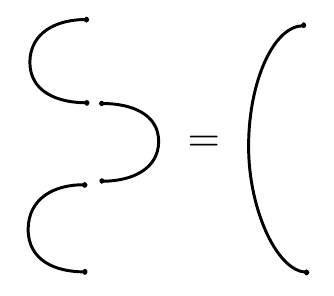}}}$ &
$
(\mathrm{id}\otimes \eta\otimes \mathrm{id})\circ (K\otimes K)=K
  $
\\
\hline
\end{tabular}
\end{equation}
which is equivalent to idempotency of the operator
\begin{equation}
 \begin{tabular}{|c|c|}
 \hline
$\vcenter{\hbox{\includegraphics[scale=1]{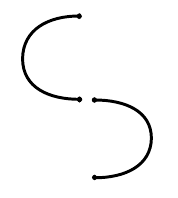}}}$ &
$
\begin{array}{c}
\pi\equiv (\mathrm{id}\otimes \eta) \circ (K\otimes \mathrm{id}):
\CH_\pt^{(0,2)} \longrightarrow \CH_\pt^{(0,2)}, \\
\\
 \pi^2=\pi
\end{array}
  $
\\
\hline
\end{tabular}
\end{equation}
It follows that $\pi$ is a projector and acts as the identity map when restricted on $ \tilde{\CH}_\pt^{(0,2)}\equiv \pi ( \CH_\pt^{(0,2)})$.

%%%%%%%%%%%%%%%%%%%%%%%%%%%%%%%%%%%%%%%%%%%%%%%%%%%%%%%%%

\section{Partial topological twisting of $\CN=2$ $d=4$ theory} \label{sec:twist}

Let us compactify 4d $\CN=2$ theory on a Riemann surface $\CC_g$ of genus $g$ without punctures and take the zero-volume limit to get a 2d theory. In order to preserve supersymmetry, we perform topological twisting along $\CC_g$ \cite{Bershadsky:1995vm}. The symmetry group of the 4d $\CN=2$ superconformal theory includes $SU(2)_L \times SU(2)_R \times SU(2)_I \times U(1)_r$, where $SU(2)_L \times SU(2)_R = SO(4)$ is the Lorentz group and $SU(2)_I \times U(1)_r$ is the R-symmetry group.
Upon dimensional reduction, the symmetry group becomes $SO(2)_E \times SO(2)_\CC \times SU(2)_I \times U(1)_r$, where $SO(2)_E$ and $SO(2)_\CC$ are the Lorentz group along the $\IR^2$ and $\CC_g$ respectively. Now, we perform topological twist along the $\CC_g$ direction. This type of twisting is studied in \cite{Kapustin:2006hi}.

\begin{table}[h]
\be
\begin{array}{c|cccc|cc|cc}
Q & SU(2)_L & SU(2)_R & SU(2)_I & U(1)_r & SO(2)_E & SO(2)_\CC & SO(2)_\CC'  & SO(2)_\CC'' \\
\hline \hline
Q_-^1 & -\half & 0 & \half & \half & - \half & -\half & 0 & 0 \\
Q_+^1& \half & 0 & \half & \half  & \half & \half & 1 & 1 \\
Q_-^2  & -\half & 0 & -\half & \half & -\half & -\half & 0 & -1\\
Q_+^2  & \half & 0 & -\half & \half & \half & \half & 1 & 0 \\
\hline
\tilde{Q}_-^1 & 0 & -\half & \half & -\half & -\half & \half & 0 & 1 \\
\tilde{Q}_+^1& 0 & \half & \half & -\half & \half & -\half & -1 & 0 \\
\tilde{Q}_-^2 & 0 & -\half & -\half & -\half & -\half & \half & 0 & 0 \\
\tilde{Q}_+^2 & 0 & \half & -\half & -\half & \half & -\half & -1 & -1
\end{array} \nn
\ee
\caption{Supercharges of the $d=4, \CN=2$ supersymmetry. Here $SO(2)_\CC'$ is the diagonal of $SO(2)_\CC \times U(1)_r$ and $SO(2)_\CC''$ is the diagonal of $SO(2)_\CC \times SU(2)_I$.}
\label{table:N2susy}
\end{table}

There are two independent choices of twisting. We can twist with either $U(1)_r$ or $SU(2)_I$. If we twist by $U(1)_r$, we get $\CN=(0, 4)$ SUSY in two-dimension since $Q_-^1, Q_-^2, \tilde{Q}_-^1, \tilde{Q}_-^2$ are preserved in 2d. Note that they all have charge $-\half$ under $SO(2)_E$. If we twist with $SU(2)_I$, the conserved supercharges are $Q_-^1, Q_+^2, \tilde{Q}_+^1, \tilde{Q}_-^2$ so that we get $\CN=(2, 2)$. See the table \ref{table:N2susy}. If we consider a linear combination of the two twists, we get $\CN=(0, 2)$ SUSY.

Let us consider twisting the free hypermultiplet and vector multiplet. We first summarize the result in the table \ref{table:twistFree} and then give a detailed account in the following.
\begin{table}[h]
\centering
\begin{tabular}{|c|cc|}
\hline
4d $\CN=2$ & $\CN=(0, 4)$ twist & $\CN=(2, 2)$ twist \\
\hline
hypermultiplet & 1 hyper, $g$ Fermi & $2 \times h^0(\CC_g, K^{\half})$ chiral \\
vector & 1 vector, $g$ twisted hyper & 1 vector, $g$ chiral  \\
\hline
\end{tabular}
\caption{Summary of the partial topological twisting of the free 4d $\CN=2$ multiplets.}
\label{table:twistFree}
\end{table}

\paragraph{$U(1)_r$ twisting}

\begin{table}[h]
\be
\begin{array}{c|cccc|cc|cc}
 & SU(2)_L & SU(2)_R & SU(2)_I & U(1)_r & SO(2)_E & SO(2)_\CC & SO(2)_\CC'  & SO(2)_\CC'' \\
\hline \hline
\psi_\pm & \pm \half & 0 & 0 & - \half & \pm \half & \pm \half & (0, -1) & \pm \half \\
\tilde{\psi}_{\dot{\pm}}^\dagger & 0 & \pm \half & 0 & \half & \pm \half & \mp \half & (0, 1) & \mp \half \\
\psi^\dagger_{\dot{\pm}} & 0 & \pm \half & 0 & \half & \pm \half & \mp \half & (0, 1) & \mp \half \\
\tilde{\psi}_{\pm} & \pm \half & 0 & 0 & -\half & \pm \half & \pm \half & (0, -1) & \pm \half \\
\hline
q & 0 & 0 & \half & 0 & 0 & 0 & 0 & \half \\
\tilde{q}^\dagger & 0 & 0 & -\half & 0 & 0 & 0 & 0 & -\half \\
q^\dagger & 0 & 0 & -\half & 0 & 0 & 0 & 0 & -\half \\
\tilde{q} & 0 & 0 & \half & 0 & 0 & 0 & 0 & \half
\end{array} \nn
\ee
\caption{Twisting hypermultiplets}
\label{table:Hyper}
\end{table}

By looking at the table \ref{table:Hyper}, we see that for the $U(1)_r$ twisting, 4 components $\psi_+, \tilde{\psi}_{+}, q, \tilde{q}$ (and its complex conjugate) form a $(0, 4)$ hypermultiplet in 2d spacetime, and also become scalar on $\CC$.
The other two components $\psi_-, \tilde{\psi}_{-}$ (along with their complex conjugates) form a $(0, 4)$ Fermi multiplet in 2d spacetime since they all become left-handed spinors. They become one-forms on $\CC$. Since $\textrm{dim}H^1 (\CC_g) = 2g$, we get $g$ (complex) Fermi multiplets in 2d.

\begin{table}[h]
\be
\begin{array}{c|cccc|cc|cc}
 & SU(2)_L & SU(2)_R & SU(2)_I & U(1)_r & SO(2)_E & SO(2)_\CC & SO(2)_\CC'  & SO(2)_\CC'' \\
\hline \hline
A_{\a \dot{\b}} & \pm \half & \pm \half & 0 & 0 & (1, -1, 0, 0) & (0, 0, 1, -1) & (0, 0, 1, -1) & (0, 0, 1, -1) \\
\lambda_\pm & \pm \half & 0 & \half & \half & \pm \half & \pm \half & (1, 0) & (1, 0) \\
\tilde{\lambda}_\pm & \pm \half & 0 & -\half & \half & \pm \half & \pm \half & (1, 0) & (0, -1) \\
\lambda^\dagger_{\dot{\pm}} & 0 & \pm \half & -\half & -\half & \pm \half & \mp \half & (-1, 0) & (-1, 0) \\
\tilde{\lambda}^\dagger_{\dot{\pm}} & 0 & \pm \half & \half & -\half & \pm \half & \mp \half & (-1, 0) & (0, 1)\\
\phi & 0 & 0 & 0 & 1 & 0 & 0 & 1 & 0 \\
\phi^\dagger & 0 & 0 & 0 & -1 & 0 & 0 & -1 & 0
\end{array} \nn
\ee
\caption{Twisting vector multiplets}
\label{table:Vector}
\end{table}

The vector multiplets, twisting with $U(1)_r$, give us 1 $(0, 4)$ vector multiplet from $A_{+\dot{+}}, \lambda_-, \tilde{\lambda}_-$ and $g$ $(0, 4)$ twisted hypermultiplets from $A_{+ \dot{-}}, \lambda_+, \tilde{\lambda}_+, \phi$ (and its complex conjugates).

Let us write the charges of the matter content for the $U(1)_r$ twist. Upon partial compactification, the $SU(2)_I$ becomes the two-dimensional $R$-symmetry $SU(2)_R$ and the twisted Lorentz group on the Riemann surface becomes a global (non-$R$) symmetry in 2d.
\begin{table}[h]
\centering
\begin{tabular}{c|cc|c|c}
 superfield & $U(1)_I \subset SU(2)_I$ & $U(1)_\CC'$ & $U(1)_r$ & components\\
 \hline \hline
 ${U}$ & $0~ (0, \half)$ & $0$ & $0~ (0, \half)$& $A_{+\dot{+}}, \lambda_-$\\
 $\Theta$ & $-\half$ & $0$ & $\half$& $\tilde{\lambda}_-$ \\
 \hline
 $\S^{(i)}$ & $0~ (0, -\half)$ & $-1$ & $0~ (0, -\half)$& $A_{-\dot{+}}, \lambda_+^\dagger$\\
 $\tilde{\S}^{(i)}$ & $0~ (0, -\half)$ & $1$ & $1~ (1, \half)$& $\phi, \tilde{\lambda}_+$\\
 \hline \hline
 $\Phi$ & $\half~ (\half, 0)$ & $0$ & $0~ (0, -\half)$ & $q, \psi_+$\\
 $\tilde{\Phi}$ & $\half~ (\half, 0)$ & $0$ & $0~ (0, -\half)$ & $\tilde{q}, \tilde{\psi}_+$ \\
 \hline
 $\Gamma^{(i)}$ & $0$ & $1$ & $\half$ & $\psi_-^\dagger$ \\
 $\tilde{\Gamma}^{(i)}$ & $0$ & $1$ & $\half$ & $\tilde{\psi}_-^\dagger$
\end{tabular}
\caption{The matter content of the $U(1)_r$ twisted free vector/hypermultiplet in terms of $\CN=(0, 2)$ superfields. $({U}, \Theta)$ form an $\CN=(0, 4)$ vector multiplet, and $(\S, \tilde{\S})$ form a twisted hypermultiplet. The superfields $(\Phi, \tilde{\Phi})$ form a hypermultiplet and $\Gamma, \tilde{\Gamma}$ are the Fermi multiplets. Here $i=1, \cdots, g$. }
\label{table:twistFreeS}
\end{table}
The components $(A_{+ \dot{+}}, \lambda_-)$ forms an vector $\CN=(0, 2)$ multiplet ${U}$, and $(\tilde{\lambda}_-)$ form a Fermi multiplet $\Theta$. The components $(A_{-\dot{+}}, \lambda_+^\dagger)$ form a chiral multiplet $\Sigma$, and $(\phi, \tilde{\lambda}_+)$ form a chiral multiplet $\tilde{\Sigma}$. We have $g$ copies of $\S, \tilde{\S}$.
Now, from the 4d hypermultiplet, we get chiral multiplets $\Phi$ and $\tilde{\Phi}$ from $(\tilde{q}, \tilde{\psi}_+)$ and $(q, \psi_+)$ respectively. We get Fermi multiplets $\Gamma, \tilde{\Gamma}$ from $\psi_-, \tilde{\psi}_-$ respectively.  We summarize this in table \ref{table:twistFreeS}.

\paragraph{$SU(2)_I$ twisting}

Let us consider the case of $SU(2)_I$ twisting. For this case, we get $\CN=(2, 2)$ supersymmetry in 2d. Now all the components of the hypermultiplets become spinors on $\CC$. We get a pair of chiral multiplets $Q=(q, \psi_+, \psi^\dagger_- )$, $\tilde{Q}=(\tilde{q}, \tilde{\psi}_+, \tilde{\psi}_{\dot{-}}^\dagger)$ in 2d, that transform as spinors on $\CC$.

When twisting the vector multiplet, we get $1$ $\CN=(2, 2)$ vector multiplet $U$ from $(A_{+ \dot{+}}, \lambda_-, \tilde{\lambda}_+, \phi)$, and $g$ $\CN=(2, 2)$ chiral multiplets $\Phi$ from $(A_{- \dot{+}}, \lambda_+^\dagger, \tilde{\lambda}_{\dot{+}})$. We summarize the matter content and charges on the table \ref{table:twistFreeI}.
\begin{table}[h]
\centering
\begin{tabular}{c|cc|c|c}
superfield & $U(1)_r \propto U(1)_A$ & $U(1)_I \propto U(1)_V$ & $SO(2)_C''$ & components \\
\hline \hline
$U$ & $0~ (0, \half, -\half, 1)$ & $0 ~ (0, \half, \half, 0)$ & $0$ & $(A_{+ \dot{+}}, \lambda_-, \tilde{\lambda}^\dagger_{\dot +}, \phi^\dagger)$ \\
$\Phi$ & $0~ (0, -\half, \half)$ & $0~ (0, -\half, -\half)$ & $1$ & $(A_{- \dot{+}}, \lambda_+^\dagger, \tilde{\lambda}_{-})$ \\
\hline
$Q$ & $0~ (0, -\half, \half)$ & $\half~ (\half, 0, 0)$ & $\half$ & $(q, \psi_+, \psi^\dagger_- )$ \\
$\tilde{Q}$ & $0~ (0, -\half, \half)$ & $\half~ (\half, 0, 0)$ & $\half$  & $(\tilde{q}, \tilde{\psi}_+, \tilde{\psi}_{\dot{-}}^\dagger)$
\end{tabular}
\caption{The matter content of the $SU(2)_I$ twisted free vector/hypermultiplets in terms of $\CN=(2, 2)$ superfields. Here $R$-charges of the superfield and components are written simultaneously. }
\label{table:twistFreeI}
\end{table}

Note that both $U(1)_r$ and $U(1)_I$ become the $R$-symmetry of the theory upon appropriate rescaling since supercharges are charged under them. We see that the vector $R$-charge is given by $R_V = 2I$ and the axial $R$-charge is given by $R_A = 2 r$, which is consistent with $\CN=(2, 2)$ superconformal symmetry. We can write left/right-moving $R$-charges to be $(J_L, J_R)=(I-r, I+r)$. Note that under this charge assignment, $\CN=(2, 2)$ supercharges $Q_-^1, Q_+^2, \tilde{Q}_{\dot +}^1, \tilde{Q}_{\dot -}^2$ have $R$-charges $(J_L, J_R) = (0, 1), (-1, 0), (1, 0), (0, -1)$.

The number of chiral multiplets of the $\CN=(2, 2)$ twist (or $SU(2)_I$ twist) is given by the number of harmonic spinors on the curve $\CC_g$ or $h^0(\CC_g, K^{\half})$. This number depends on the choice of spin structure on $\CC_g$ \cite{MR0358873}.

\bibliographystyle{jhep}
\bibliography{refs}

\end{document}